\documentclass[a4paper,UKenglish,cleveref, autoref, thm-restate]{lipics-v2021}
\newboolean{arXiv}
\setboolean{arXiv}{true}

\newboolean{final}
\setboolean{final}{true}

\newboolean{debug}
\setboolean{debug}{false}

\ifthenelse{\boolean{arXiv}}{
    \pdfoutput=1 %
    \hideLIPIcs  %
}{}

\usepackage{CJKutf8}

\usepackage{environ}
\usepackage{hyperref}
\ifLuaTeX
\usepackage{newunicodechar}
\newunicodechar{–}{--}
\fi
\usepackage{environ}
\NewEnviron{myprotect}{\BODY}

\usepackage{ebproof}
\usepackage{stmaryrd} %
\usepackage{stackengine} %

\usepackage{tikz}
\usetikzlibrary{shapes, arrows, arrows.spaced, arrows.meta, chains,positioning, cd, shapes.geometric,
    decorations, decorations.pathmorphing, decorations.markings, decorations.pathreplacing, automata, backgrounds, petri, matrix, fit, calc, graphs, quotes, bending}

\usepackage{amssymb}
\usepackage{mathtools}
\usepackage{thmtools}

\DeclarePairedDelimiterX\tog[1]{\langle\!\langle}{\rangle\!\rangle}{#1}

\makeatletter
\let\orgdescriptionlabel\descriptionlabel
\renewcommand*{\descriptionlabel}[1]{%
  \let\orglabel\label
  \let\label\@gobble
  \phantomsection
  \edef\@currentlabel{#1}%
  \let\label\orglabel
  \orgdescriptionlabel{#1}%
}
\makeatother

\makeatletter
\newcommand{\customlabel}[2]{%
   \protected@write \@auxout {}{\string \newlabel {#1}{{#2}{\thepage}{#2}{#1}{}} }%
   \hypertarget{#1}{#2}
}
\makeatother

\crefname{thm}{Thm.}{Thms.}
\crefname{prop}{Prop.}{Props.}
\crefname{property}{Property}{Properties}
\crefname{lem}{Lem.}{Lems.}
\crefname{cor}{Cor.}{Cors.}
\crefname{defi}{Def.}{Defs.}
\crefname{section}{\S}{\S}
\crefname{figure}{Fig.}{Figs.}
\crefname{exa}{Example}{Examples}
\crefname{rem}{Remark}{Remarks}

\usepackage{centernot}

\newcommand{\proofcasebase}[1]{\colorbox{black!10}{\hspace{.2em}\textsf{#1}\hspace{.2em}}~}
\newcommand{\proofcase}[1]{\noindent\setlength{\fboxsep}{0pt}\proofcasebase{#1}}

\newtheorem*{theorem*}{Theorem}

\definecolor[named]{ACMBlue}{cmyk}{1,0.1,0,0.1}
\definecolor[named]{ACMYellow}{cmyk}{0,0.16,1,0}
\definecolor[named]{ACMOrange}{cmyk}{0,0.42,1,0.01}
\definecolor[named]{ACMRed}{cmyk}{0,0.90,0.86,0}
\definecolor[named]{ACMLightBlue}{cmyk}{0.49,0.01,0,0}
\definecolor[named]{ACMGreen}{cmyk}{0.20,0,1,0.19}
\definecolor[named]{ACMPurple}{cmyk}{0.55,1,0,0.15}
\definecolor[named]{ACMDarkBlue}{cmyk}{1,0.58,0,0.21}
\hypersetup{colorlinks=true,breaklinks=true,
    linkcolor=ACMPurple,
    citecolor=ACMPurple,
    urlcolor=ACMDarkBlue,
    filecolor=ACMDarkBlue}

\ifthenelse{\boolean{debug}}{%
\usepackage[xcolor, cleveref, notion, quotation, composition]{knowledge}
\knowledgeconfigure{notion}
\knowledgestyle {notion}{color=green!70!black!100}
\knowledgestyle {intro}{color=blue}
\knowledgestyle {kl unknown}{color=red}
\knowledgestyle {kl unknown cont}{color=red}
}{
\usepackage[xcolor, cleveref, notion, quotation, electronic]{knowledge}
\knowledgeconfigure{notion}
\definecolor{Dark Ruby Red}{HTML}{580507}
\definecolor{Dark Blue Sapphire}{HTML}{053641}
\knowledgestyle{intro notion}{color={Dark Ruby Red}, emphasize}
\knowledgestyle{notion}{color={Dark Blue Sapphire}}
}

\usepackage{mathcommand}
\knowledgeconfigure {diagnose bar=true, diagnose line=true}\knowledgeconfigure{protect quotation={graph,figure,tikzpicture,tikzcd}}
\NewEnviron{hidden}{\BODY} %
\ifthenelse{\boolean{final}}{%
\usepackage[backgroundcolor=orange!20, textsize=tiny, disable]{todonotes}
}{
\usepackage[backgroundcolor=orange!20, textsize=tiny]{todonotes}
}
\definecolor{yoshikieditcolor}{RGB}{200,230,200}

\NewEnviron{sideyoshiki}{\todo[caption={},backgroundcolor=yoshikieditcolor, size=\tiny]{YN: \BODY}}
\NewEnviron{yoshiki}{\todo[inline,caption={},backgroundcolor=yoshikieditcolor, size=\footnotesize]{YN: \BODY}} %
 \newtheorem{problem}[theorem]{Problem}
\NewEnviron{lem}{\begin{lemma}\BODY\end{lemma}}
\NewEnviron{cor}{\begin{corollary}\BODY\end{corollary}}
\NewEnviron{defi}{\begin{definition}\BODY\end{definition}}
\NewEnviron{exa}{\begin{example}\BODY\end{example}}
\NewEnviron{rem}{\begin{remark}\BODY\end{remark}}
\NewEnviron{prob}{\begin{problem}\BODY\end{problem}}

\DeclarePairedDelimiter\set{\{}{\}}
\DeclarePairedDelimiter\tuple{(}{)}

\knowledgenewrobustcmd{\range}[2]{\cmdkl{[}#1\,\cmdkl{..}\,#2\cmdkl{]}}
\knowledgenewrobustcmd{\rangeone}[1]{\cmdkl{[}#1\cmdkl{]}}

\knowledgenewrobustcmd{\getith}[1]{\cmdkl{[}#1\cmdkl{]}}

\knowledgenewrobustcmd{\card}{\cmdkl{\#}}
\knowledgenewrobustcmd{\pset}{\mathord{\cmdkl{\wp}}}
\knowledgenewrobustcmd{\dcup}{\mathop{\cmdkl{\sqcup}}}
\newcommand{\const}[1]{\mathsf{#1}}
\newcommand{\bl}{\cdot}
\newcommand{\defeq}{\coloneq}

\newcommand{\nat}{\mathbb{N}}

\knowledgenewrobustcmd{\univ}[1]{\cmdkl{|}#1\cmdkl{|}}

\knowledgenewrobustcmd{\diagonal}{\cmdkl{\triangle}}
\knowledgenewrobustcmd{\quo}[1]{\cmdkl{[}#1\cmdkl{]}}

\knowledgenewrobustcmd{\bigO}{\cmdkl{O}}

\knowledgenewrobustcmd{\klmodels}{\mathrel{\cmdkl{\models}}}

\knowledgenewrobustcmd{\autolen}[1]{\cmdkl{\|}#1\cmdkl{\|}}

\knowledgenewrobustcmd{\ljump}[1]{\cmdkl{[}#1\cmdkl{]}}

\NewDocumentCommand\word{O{1}}{%
    \ifcase#1
        undefined
    \or w
    \or v
    \or u
    \else undefined

    \fi
}

\NewDocumentCommand\la{O{1}}{%
    \ifcase#1
        undefined
    \or \mathcal{L}
    \or \mathcal{K}
    \else undefined

    \fi
}

\knowledgenewrobustcmd{\laInac}{\cmdkl{\mathcal{L}}^{\cmdkl{\mathrm{Inac}}}}
\knowledgenewrobustcmd{\laNorm}{\cmdkl{\mathcal{L}}^{\cmdkl{\mathrm{Norm}}}}
\knowledgenewrobustcmd{\laIncon}{\cmdkl{\mathcal{L}}^{\cmdkl{\mathrm{Incon}}}}
\knowledgenewrobustcmd{\laTop}{\cmdkl{\mathcal{L}}^{\cmdkl{\mathrm{Top}}}}
\knowledgenewrobustcmd{\laTest}{\cmdkl{\mathcal{L}}^{\cmdkl{\mathrm{Test}}}}
\knowledgenewrobustcmd{\laConv}{\cmdkl{\mathcal{L}}^{\cmdkl{\mathrm{Conv}}}}
\knowledgenewrobustcmd{\laNom}{\cmdkl{\mathcal{L}}^{\cmdkl{\mathrm{Nom}}}}

\NewDocumentCommand\rel{O{1}}{%
    \ifcase#1
        undefined
    \or R
    \or S
    \else undefined

    \fi
}

\knowledgenewrobustcmd{\resTop}{\cmdkl{\mathrm{Top}}}
\knowledgenewrobustcmd{\resTest}{\cmdkl{\mathrm{Test}}}
\knowledgenewrobustcmd{\resConv}{\cmdkl{\mathrm{Conv}}}
\knowledgenewrobustcmd{\resNom}{\cmdkl{\mathrm{Nom}}}

\newcommand{\mul}[1]{\bar{#1}}
\NewDocumentCommand\struc{O{1}}{%
    \ifcase#1
        undefined
    \or \mathcal{M}
    \or \mathcal{N}
    \else undefined

    \fi
}

\NewDocumentCommand\graph{O{1}}{%
    \ifcase#1
        undefined
    \or G
    \or H
    \or J
    \else undefined

    \fi
}
\knowledgenewrobustcmd{\eps}{\cmdkl{\varepsilon}}

\knowledgenewrobustcmd{\src}{\cmdkl{\const{1}}}
\knowledgenewrobustcmd{\tgt}{\cmdkl{\const{2}}}

\knowledgenewrobustcmd{\autotrans}{\cmdkl{\delta}}
\knowledgenewrobustcmd{\autosrc}{\cmdkl{\const{1}}}

\knowledgenewrobustcmd{\Bautosrc}{\cmdkl{\const{1}}}

\NewDocumentCommand\aterm{O{1}}{%
    \ifcase#1
        undefined
    \or a
    \or b
    \or c
    \or d
    \else undefined

    \fi
}
\NewDocumentCommand\term{O{1}}{%
    \ifcase#1
        undefined
    \or t
    \or s
    \or {u}
    \else undefined

    \fi
}

\NewDocumentCommand\fml{O{1}}{%
    \ifcase#1
        undefined
    \or \varphi
    \or \psi
    \or \rho
    \else undefined

    \fi
}

\newcommand{\union}{\mathbin{+}}
\newcommand{\intersection}{\mathbin{\cap}}
\newcommand{\id}{\const{I}}
\newcommand{\emp}{\const{0}}
\newcommand{\compo}{\mathbin{;}}

\newcommand{\dom}{\mathsf{d}}
\newcommand{\rng}{\mathsf{r}}
\newcommand{\lop}{\circlearrowleft}

\NewDocumentCommand\tset{O{1}}{%
    \ifcase#1
        undefined
    \or T
    \or S
    \or {U}
    \else undefined

    \fi
}

\NewDocumentCommand\lab{O{1}}{%
    \ifcase#1
        undefined
    \or x
    \or y
    \or z
    \else undefined

    \fi
}

\NewDocumentCommand\glang{O{1}}{%
    \ifcase#1
        undefined
    \or \mathcal{G}
    \or \mathcal{H}
    \else undefined

    \fi
}

\NewDocumentCommand\algclass{O{1}}{%
    \ifcase#1
        undefined
    \or \mathcal{C}
    \or \mathcal{D}
    \else undefined

    \fi
}

\knowledgenewrobustcmd{\REL}{\mathsf{\cmdkl{REL}}}
\knowledgenewrobustcmd{\RELk}[1]{\mathsf{\cmdkl{REL}}_{#1}}
\knowledgenewrobustcmd{\LANG}{\mathsf{\cmdkl{LANG}}}

\NewDocumentCommand\trace{O{1}}{%
    \ifcase#1
        undefined
    \or \tau
    \or \sigma
    \or \rho
    \else undefined
    \fi
}

\tikzset{
    every edge/.append style = {
            line width = .3pt,
        },
    plab/.style={line width = 0.1pt, fill=#1, inner sep = .025cm, anchor=center, font = \fontsize{6pt}{0}},
    plab/.default= white,
    elab/.style={draw, rectangle, line width = 0.1pt, fill=#1, inner sep = .035cm, anchor=center, font = \footnotesize},
    elab/.default= white,
    tlab/.style={line width = 0.1pt, fill=#1, inner sep = .025cm, anchor=center, font = \fontsize{6pt}{0}\selectfont},
    tlab/.default= white
}
\tikzset{
    png export/.style={
            external/system call/.add={}{; convert -density 300 -transparent white "\image.pdf" "\image.png"},
            /pgf/images/external info,
            /pgf/images/include external/.code={
                    \includegraphics[width=\pgfexternalwidth,height=\pgfexternalheight]{##1.png}
                },
        }
}

\tikzstyle{mynode} = [inner sep = 1.5pt, fill= gray!20, font=\footnotesize]
\tikzstyle{mysmallnode} = [inner sep = 1.pt, fill= gray!20]

\tikzstyle{vert} = [draw, circle, mynode]

\tikzset{earrow/.style={>={{[flex] Latex[length=.1cm, width=2.5pt]}}}}

\tikzset{homoarrow/.style={earrow, line width = .3pt, color = olive, opacity=0.8}}

\tikzstyle{elabel} = [inner sep = 1.pt, font = \scriptsize, opacity = 1]

\tikzstyle{dvert} = [vert, color = gray!20]
\tikzstyle{dearrow} = [earrow, color = gray!20]

\knowledgenewrobustcmd{\jump}[1]{\cmdkl{\llbracket} #1 \cmdkl{\rrbracket}}

\newcommand{\vsig}{\Sigma}

\NewDocumentCommand\automaton{O{1}}{%
    \ifcase#1
        undefined
    \or \mathcal{A}
    \or \mathcal{B}
    \or \mathcal{C}
    \else undefined
    \fi
}

\knowledgenewrobustcmd{\iw}{\cmdkl{\operatorname{iw}}}
\knowledgenewrobustcmd{\pw}{\cmdkl{\operatorname{pw}}}
\knowledgenewrobustcmd{\tw}{\cmdkl{\operatorname{tw}}}

\knowledgenewrobustcmd{\loopRKAautomaton}[1]{\cmdkl{\automaton}_{#1}}
\knowledgenewrobustcmd{\binloopRKAautomaton}[1]{\cmdkl{\automaton}_{#1}}

\knowledgenewrobustcmd{\loopautomatonAFA}{\cmdkl{\automaton[2]}}
\knowledgenewrobustcmd{\binloopautomatonAFA}{\cmdkl{\tilde{\automaton[2]}}}

\knowledgenewrobustcmd{\LoopLabel}[1]{\cmdkl{L}_{#1}}

\knowledgenewrobustcmd{\transrel}[2]{\mathrel{\cmdkl{\longrightarrow}}^{#1}_{#2}}
\knowledgenewrobustcmd{\decomposedtransrel}[2]{\mathrel{\cmdkl{\longrightarrow}}^{#1}_{#2}}

\knowledgenewrobustcmd{\univl}[3]{\cmdkl{|}#1\cmdkl{|}_{\cmdkl{(}#2\cmdkl{<)}}^{#3}}
\knowledgenewrobustcmd{\univc}[3]{\cmdkl{|}#1\cmdkl{|}_{\cmdkl{(}#2\cmdkl{)}}^{#3}}
\knowledgenewrobustcmd{\univr}[3]{\cmdkl{|}#1\cmdkl{|}_{\cmdkl{(}#2\cmdkl{>)}}^{#3}}

\knowledgenewrobustcmd{\PBFML}{\cmdkl{\mathbb{B}_{+}}}

\knowledgenewrobustcmd{\glue}{{\cmdkl{\odot}}}
\knowledgenewrobustcmd{\bigglue}{{\cmdkl{\bigodot}}}

\newcommand{\series}{\diamond}

\knowledgenewrobustcmd{\lanchor}{{\cmdkl{\triangleright}}}
\knowledgenewrobustcmd{\ranchor}{{\cmdkl{\triangleleft}}}

\knowledgenewrobustcmd{\binenc}[2]{\cmdkl{\lfloor} #1 \cmdkl{\rfloor}_{#2}} %

\knowledge{text = \textsc{PSpace}}
| PSpace
| PSPACE

\knowledge{text = \textsc{ExpTime}}
| ExpTime
| EXPTIME

\knowledge{text = \textsc{ExpSpace}}
| ExpSpace
| EXPSPACE

\knowledge{text={i.e.\@}, italic}
  | ie

\knowledge{text={I.e.\@}, italic}
  | Ie

\knowledge{text={s.t.\@}, italic}
  | st

\knowledge{text={e.g.\@}, italic}
  | eg

\knowledge{text={E.g.\@}, italic}
  | Eg

\knowledge{text={vs.\@}, italic}
  | vs

\knowledge{text={w.r.t.\@}, italic}
  | wrt

\knowledge{text={a.k.a.\@}, italic}
  | aka

\knowledge{text={w.l.o.g.\@}, italic}
  | wlog

\knowledge{text={W.l.o.g.\@}, italic}
  | Wlog

\knowledge{text={cf.\@}, italic}
  | cf

\knowledge{text={Cf.\@}, italic}
  | Cf

\knowledge{text={iff}, italic}
  | iff

\knowledge{text={r.e.\@}, italic}
  | r.e.
  | re

\knowledge{text={resp.\@}, italic}
  | resp

\knowledge{notion}
| string
| strings

\knowledge{notion}
| empty string

\knowledge{notion}
| character
| characters

\knowledge{notion}
| concatenation@word

\knowledge{notion}
| concatenation@lang

\knowledge{notion}
| Kleene star

\knowledge{notion}
| Kleene plus

\knowledge{notion}
| string language
| string languages
| language
| languages

\knowledge{notion}
| language inclusion problem

\knowledge{notion}
| language equivalence problem

\knowledge{notion}
| regular expressions

\knowledge{notion}
| 1NFA
| 1NFAs
| nondeterministic string automaton
| nondeterministic string automata

\knowledge{notion}
| string automata@2AFA
| 2AFA
| 2AFAs
| 2-way alternating string automaton
| 2-way alternating string automata
| 2-way alternating finite string automaton
| 2-way alternating finite string automata

\knowledge{notion}
| state
| states

\knowledge{notion}
| run
| runs

\knowledge{notion}
| binary relation
| binary relations

\knowledge{notion}
| identity relation

\knowledge{notion}
| empty relation 

\knowledge{notion}
| union

\knowledge{notion}
| intersection

\knowledge{notion}
| converse

\knowledge{notion}
| composition

\knowledge{notion}
| $n$-th iteration

\knowledge{notion, scope = lang}
| $n$-th iteration

\knowledge{notion}
| transitive closure

\knowledge{notion}
| reflexive transitive closure

\knowledge{notion}
| domain

\knowledge{notion}
| range

\knowledge{notion}
| antidomain

\knowledge{notion}
| antirange

\knowledge{notion}
| top

\knowledge{notion}
| graph loop
| graph loops
| strong loop predicate

\knowledge{notion}
| graph
| graphs

\knowledge{notion}
| vertex
| vertices

\knowledge{notion}
| edge
| edges

\knowledge{notion}
| source
| sources

\knowledge{notion}
| target
| targets

\knowledge{notion}
| empty

\knowledge{notion}
| null

\knowledge{notion}
| path graph
| path graphs

\knowledge{notion}
| quotient graph
| quotient graphs

\knowledge{notion}
| graph homomorphism
| graph homomorphisms

\knowledge{notion}
| graph isomorphism
| graph isomorphisms
| graph isomorphic

\knowledge{notion}
| graph isomorphism closure

\knowledge{notion}
| path decomposition
| path decompositions

\knowledge{notion, scope = pathwidth}
| width
| widths

\knowledge{notion}
| pathwidth 
| pathwidths

\knowledge{notion}
| atomic graph

\knowledge{notion}
| atomic path decomposition

\knowledge{notion}
| term@loop-RKA
| terms@loop-RKA
| loop-RKA term
| loop-RKA terms

\knowledge{notion}
| variable
| variables

\knowledge{notion}
| fresh

\knowledge{notion}
| atomic test
| atomic tests

\knowledge{notion}
| test
| tests

\knowledge{notion}
| nominal
| nominals

\knowledge{notion}
| constant
| constants

\knowledge{notion}
| positive boolean formula
| positive boolean formulas

\knowledge{notion}
| propositional variable
| propositional variables

\knowledge{notion}
| relational Kleene algebra
| relational KA
| RKA
| Relational Kleene algebra

\knowledge{notion}
| regular path queries
| RPQs
| regular path query
| RPQ

\knowledge{notion}
| relational Kleene algebra with tests
| relational KAT
| RKAT

\knowledge{notion}
| KAT term
| RKAT term
| term@KAT

\knowledge{notion}
| relational Kleene lattice

\knowledge{notion}
| PCoR*
| the positive calculus of relations with transitive closure

\knowledge{notion}
| RKL derivatives
| derivatives@RKL
| derivative@RKL

\knowledge{notion}
| series-parallel graph
| series-parallel graphs

\knowledge{notion}
| loop-RKA
| Loop-RKA
| relational Kleene algebra with graph loop

\knowledge{notion}
| loop-RKAT

\knowledge{notion}
| KAT
| Kleene algebra with tests

\knowledge{notion}
| variable complements

\knowledge{notion}
| constant complements

\knowledge{notion}
| PDL

\knowledge{notion}
| loop-PDL

\knowledge{notion}
| loop-automaton
| loop-automata

\knowledge{notion}
| Kleene algebra
| Kleene algebras
| KA
| KAs

\knowledge{notion}
| Kleene algebra term
| Kleene algebra terms
| KA term
| KA terms
| RKA term
| RKA terms

\knowledge{notion}
| equation
| equations

\knowledge{notion}
| inequation
| inequations

\knowledge{notion}
| size
| sizes

\knowledge{notion}
| intersection width
| intersection widths

\knowledge{notion}
| semantics
| relational semantics

\knowledge{notion}
| structure
| structures
| (relational) structure
| (relational) structures

\knowledge{notion}
| theory

\knowledge{notion}
| equational theory
| equational theories
| inequational theory

\knowledge{notion}
| linearly bounded pathwidth model property

\knowledge{notion}
| one-way deterministic finite word automaton
| one-way deterministic finite word automata
| 1DFA
| 1DFAs

\knowledge{notion}
| trace
| traces

\knowledge{notion}
| disjoint union

\knowledge{notion}
| size@2AFA
| sizes@2AFA

\knowledge{notion}
| quotient structure

\knowledge{notion}
| coalesced product

\knowledge{notion}
| gluing operator

\knowledge{notion}
| language model
| language models

\knowledge{notion}
| relational model
| relational models

\knowledge{notion}
| incorrectness logic

\knowledge{notion}
| closed@loopA

\knowledge{notion}
| rich test
| rich tests

\knowledge{notion}
| element
| elements

\knowledge{notion}
| bag

\knowledge{notion}
| ports
| port
 
\title{The Equational Theory of Relational Kleene Algebra with Graph Loop is PSPACE-Complete}

\titlerunning{The Equational Theory of Loop-RKA is PSPACE-Complete}

\author{Yoshiki Nakamura}{Chiba University, Japan}{nakamura.yoshiki.ny@gmail.com}{https://orcid.org/0000-0003-4106-0408}{}%

\authorrunning{Y. Nakamura} %

\Copyright{Yoshiki Nakamura} %

\ccsdesc[100]{Theory of computation~Logic}
\ccsdesc[100]{Theory of computation~Formal languages and automata theory}

\keywords{Kleene algebra, Graph loop, Domain}

\category{} %

\ifthenelse{\boolean{arXiv}}{
\relatedversion{This is an extended version of the FSCD paper with the same title \cite{thispaper}.} 
}{
\relatedversion{}
\relatedversiondetails{Full Version}{https://arxiv.org/abs/2512.22930} %
}

\acknowledgements{
This work was supported by JSPS KAKENHI Grant Numbers JP21K13828 and JP25K14985.} %

\ifthenelse{\boolean{final}}{
\nolinenumbers %
}{}

\EventEditors{Frank Pfenning}
\EventNoEds{1}
\EventLongTitle{11th International Conference on Formal Structures for Computation and Deduction (FSCD 2026)}
\EventShortTitle{FSCD 2026}
\EventAcronym{FSCD}
\EventYear{2026}
\EventDate{July 20--23, 2026}
\EventLocation{Lisbon, Portugal}
\EventLogo{}
\SeriesVolume{378}
\ArticleNo{25}

\begin{document}

\maketitle

\begin{abstract}
In this paper, we show that the equational theory of relational Kleene algebra with the \emph{graph loop} operator (a.k.a.~\emph{fixset}) is \textsc{PSpace}-complete.
Here, the graph loop is the unary operator that restricts a binary relation to the identity relation.
We further show that this \textsc{PSpace}-completeness still holds by extending the terms with top, tests, converse, and nominals, over relational models.
Notably, for Kleene algebra with tests (KAT),
while the equational theory of relational KAT with antidomain is \textsc{ExpTime}-complete,
we show that the equational theory of relational KAT with domain is \textsc{PSpace}-complete,
thereby resolving a problem left open in previous works.

To this end,
we introduce a novel automaton model on relational structures (graphs), called \emph{loop-automata}.
Loop-automata extend nondeterministic finite automata with a transition type that tests whether the current vertex has a loop.
Using this model, we can give a polynomial-time reduction from the equational theories above to the language inclusion problem for 2-way alternating automata.

 \end{abstract}

\nointro{language equivalence problem}
\nointro{regular expressions}
\nointro{1NFA}
\nointro{empty relation}
\nointro{union}
\nointro{intersection}
\nointro{top}
\nointro{source}
\nointro{target}
\nointro{path graph}
\nointro{fresh}
\nointro{constant}
\nointro{KAT term}
\nointro{KAT}
\nointro{variable complements}
\nointro{PDL}
\nointro{theory}
\nointro{disjoint union}
\nointro{language model}
\nointro{relational model}
\nointro{element}

\begin{scope}

\section{Introduction} \label{section: introduction}
\AP
\intro*\kl{Relational Kleene algebra} (\reintro*\kl{RKA}), \intro*\kl{Kleene algebra} over "relational models",
is an algebraic system on binary relations with the operators $\set{\emp, \id, \compo, \union, \bl^{*}}$ of
\kl{empty relation} ($\emp$),
\kl{identity relation} ($\id$),
\kl{union} ($\union$),
\kl{composition} ($\compo$),
and \kl{reflexive transitive closure} (Kleene star) ($\bl^{*}$).
\kl{RKA} arises in various research areas,
for instance,
in database theory via \intro*\kl{regular path queries} (\kl{RPQs}) \cite{calvaneseRewritingRegularExpressions1999} and
in program verification via \intro*\kl{relational Kleene algebra with tests} (\reintro*\kl{relational KAT}, \reintro*\kl{RKAT}) \cite{kozenKleeneAlgebraTests1996}.
The \kl{equational theory} of \kl{RKA} is well-known to be decidable and \kl{PSPACE}-complete,
because it coincides\footnote{\label{footnote: lang and REL equiv}%
This coincidence can be seen, "eg", from the completeness theorem of \kl{Kleene algebra} \cite{kozenCompletenessTheoremKleene1994} or from a straightforward argument regarding labelled paths in "(relational) structures".} with the \kl{language equivalence problem} of \kl{regular expressions} (which is \kl{PSPACE}-complete \cite{meyerEquivalenceProblemRegular1972}).
The \kl{equational theory} over "relational models" is still \kl{PSPACE}-complete even if
we extend \kl{RKA} terms with several operators,
such as
\kl{tests} \cite{kozenKleeneAlgebraTests1996},
\kl{converse} \cite{brunetKleeneAlgebraConverse2014},
\kl{top} \cite{nakamuraExistentialCalculiRelations2023,pousCompletenessTheoremsKleene2024},
\kl{nominals} \cite{nakamuraDerivativesGraphsPositive2025},
and all the combinations of them \cite[Cor.~6.7]{nakamuraDerivativesGraphsPositive2025}.

A remarkable extension is \kl{RKA} with \reintro*\kl{domain} ($\bl^{\mathsf{d}}$) and \intro*\kl{antidomain} ($\bl^{\mathsf{a}}$) operators \cite{desharnaisKleeneAlgebraDomain2006,desharnaisModalSemiringsRevisited2008,desharnaisInternalAxiomsDomain2011}.
The \kl{domain} operator $\bl^{\mathsf{d}}$ maps a binary relation $R$ (on a set $X$) to the diagonal of the \kl{domain}
and the \kl{antidomain} operator $\bl^{\mathsf{a}}$ maps $R$ to the (unary) complement of $R^{\mathsf{d}}$:
\begin{align*}
  R^{\mathsf{d}} &\defeq \set{ \tuple{x,x} \mid x \in X,~ \exists y \in X. \tuple{x,y} \in R}, &
  R^{\mathsf{a}} &\defeq \set{ \tuple{x,x} \mid x \in X,~ \forall y \in X. \tuple{x,y} \not\in R}.
\end{align*}
The \kl{equational theory} of \kl{RKA} with \kl{antidomain}%
\footnote{\label{footnote: undecidable}%
The \kl{domain} $\bl^{\dom}$ can be expressed using \kl{antidomain} as $R^{\mathsf{d}} = (R^{\mathsf{a}})^{\mathsf{a}}$.
The converse does not hold in general.
Using complement $R^{-} \defeq X^2 \setminus R$ where $X$ is the base set, including top,
the antidomain $\bl^{\mathsf{a}}$ can be expressed as $R^{\mathsf{a}} = ((R^{\mathsf{d}} \top)^{-})^{\mathsf{d}}$.
However, the \kl{equational theory} of \kl{RKA} with complement is undecidable (by \cite{hardinComplexityHornTheory2003}).
It is still $\Pi^{0}_{1}$-complete even when the application of complement is restricted to \kl{constants} \cite[Corollary 7.2]{nakamuraUndecidabilityEmptinessProblem2025} or \kl{variables} \cite[Theorem 50]{nakamuraExistentialCalculiRelations2023} (via the reduction of \cite[Example 3.4]{nakamuraUndecidabilityEmptinessProblem2025} from Horn formulas to "equations").}
is \kl{EXPTIME}-complete \cite{sedlarComplexityKleeneAlgebra2023},
coinciding with the complexity of the \kl{theory} of Propositional Dynamic Logic (\kl{PDL}).
However, for the \kl{equational theory} of \kl{RKA} with \kl{domain}, the complexity was left open by McLean \cite[Problem 7.3]{mcleanFreeKleeneAlgebras2020} and Sedl{\'a}r \cite[p.~16]{sedlarKleeneAlgebraDynamic2023}.
To the best of our knowledge,%
\footnote{In \cite[unpublished manuscript]{reutterContainmentNestedRegular2013} 
it is claimed that the containment problem of \emph{nested regular expressions} \cite{perezNSPARQLNavigationalLanguage2008} ("ie", the \kl{inequational theory} of \kl{RKA} with domain and converse) is in \kl{PSPACE},
which also implies the \kl{PSPACE} upper bound of the \kl{equational theory} of \kl{RKA} with \kl{domain}.
However, in the version of \cite{reutterContainmentNestedRegular2013}, the containment on ``semipaths'' (not general "structures") is only considered as noted in \cite[p.~1]{reutterContainmentNestedRegular2013}.}
it was previously known only that the problem is \kl{PSPACE}-hard \cite{meyerEquivalenceProblemRegular1972} and in \kl{EXPTIME} \cite{sedlarComplexityKleeneAlgebra2023}.

Another extension is \kl{RKA} with \emph{\kl{intersection}} ($\intersection$), a.k.a.\ (representable) \emph{relational Kleene lattice} \cite{andrekaEquationalTheoryKleene2011}.
The \kl{equational theory} is \kl{EXPSPACE}-complete \cite{nakamuraPartialDerivativesGraphs2017,brunetPetriAutomata2017,nakamuraDerivativesGraphsPositive2025},
while it is \kl{PSPACE}-complete if we fix the \kl{intersection width} ("ie", the maximum nesting depth of $\intersection$) of input \kl(loop-RKA){terms} \cite{nakamuraDerivativesGraphsPositive2025}.
\kl{RKA} with \emph{\kl{graph loop}} ($\bl^{\lop}$) --henceforth, \emph{\kl{loop-RKA}} for short-- \cite{nakamuraDerivativesGraphsPositive2025} is a restriction of \kl{RKA} with \kl{intersection},
where we only allow the \kl{intersection} with \kl{identity relation} ("ie", the \kl{graph loop} operator).
Here, \kl{graph loop}\footnote{In \cite{daneckiPropositionalDynamicLogic1984}, the \emph{strong loop predicate} $\mathrm{loop}(\term)$ is considered. In our paper, we use it as operator.
Precisely, $\term^{\lop}$ is equal to $\mathrm{loop}(\term)?$ where $\fml?$ denotes the ""rich test"".}
\cite{daneckiPropositionalDynamicLogic1984,nakamuraUndecidabilityPositiveCalculus2024} (also called \emph{fixset} \cite{hirschAlgebraFunctionsAntidomain2016}) is the unary operator of restricting a given binary relation $R$ on $X$ to the \kl{identity relation}:
$R^{\lop} \defeq R \cap \diagonal_{X}$.
\kl{Loop-RKA} is still expressive enough to define \kl{domain} and \kl{range} operators using top ($\top$), the full relation, as follows:
\begin{align*}
\term^{\mathsf{d}} ~&=~ (\term \compo \top)^{\lop}, &
\term^{\mathsf{r}} ~&=~ (\top \compo \term)^{\lop}.
\tag{$\mathsf{d}$/$\mathsf{r}$-by-$\lop$,$\top$}\label{equation: dom range by loop top}
\end{align*}
However, the complexity of the "equational theory" of \kl{loop-RKA} (without $\top$) was also left open by Nakamura \cite[\S 8]{nakamuraDerivativesGraphsPositive2025}.
To the best of our knowledge, it was previously known only that the problem is \kl{PSPACE}-hard \cite{meyerEquivalenceProblemRegular1972} and in \kl{EXPTIME} \cite{daneckiPropositionalDynamicLogic1984,gollerPDLIntersectionConverse2009}.
\AP
\intro*\kl[loop-PDL]{Loop-PDL} \cite{daneckiPropositionalDynamicLogic1984} is propositional dynamic logic of regular programs with "graph loop" ("strong loop predicate").
The theory of \kl{loop-PDL} is \kl{EXPTIME}-complete \cite{daneckiPropositionalDynamicLogic1984} (another proof is in \cite[Corollary 4.9]{gollerPDLIntersectionConverse2009}).
By a canonical embedding ("cf", \cite[p.\ 209]{fischerPropositionalDynamicLogic1979}),
we can give a polynomial-time reduction from the \kl{equational theory} of \kl{loop-RKA} into the \kl{theory} of \kl{loop-PDL}.
The \kl{EXPTIME} upper bound for \kl{loop-RKA} can be obtained via this reduction.
For \kl{loop-PDL},
the \kl{EXPTIME} upper bound can also be obtained from the \kl{theory} of \kl{PDL} with \kl{intersection} where the "intersection width" is fixed \cite[Corollary 4.9]{gollerPDLIntersectionConverse2009}.
In contrast, we cannot obtain the \kl{PSPACE} upper bound of \kl{loop-RKA} from the \kl{equational theory} of \kl{RKA} with \kl{intersection} where the "intersection width" is fixed,
because ``"rich tests"'' in the context of \kl{PDL} are not employed in \kl{RKA} with \kl{intersection}.\footnote{%
For \kl{loop-PDL}, we can reduce the "intersection width" to at most $2$ \cite[Corollary 4.9]{gollerPDLIntersectionConverse2009}, "eg", by translating the nested $(a b^{\lop} a)^{\lop}$ into $(a (\langle b^{\lop}\rangle\mathsf{true})? a)^{\lop}$,
where $\fml?$ denotes the "rich test" and $\langle\term\rangle\fml$ denotes the diamond modality.
Note that the "intersection width" of $\fml?$ is always defined to be $1$, even if $\fml$ contains arbitrary nested "graph loop".}

\subparagraph*{Contributions}
We first show the following, which affirmatively answers \cite[\S 8]{nakamuraDerivativesGraphsPositive2025}:
\begin{theorem}\label{theorem: KL PSPACE-complete}
The \kl{equational theory} for \kl{loop-RKA} is \textup{\kl{PSPACE}}-complete.
\end{theorem}
Moreover, by adapting the encodings of some operators from \cite[\S 6]{nakamuraDerivativesGraphsPositive2025} in our automata construction,
the \kl{PSPACE} upper bound is strengthened as follows (\Cref{section: encoding extras}).
Here, "tests" are in the context of \kl{KAT} \cite{kozenKleeneAlgebraTests1996} and "nominals" are in the context of hybrid modal logic \cite{arecesHybridLogics2007}.
\begin{theorem}\gdef\theoremKLextrasPSPACEcomplete{
The \kl{equational theory}
for \kl{loop-RKA} with \kl{top}, \kl{tests}, \kl{converse}, and \kl{nominals} is \textup{\kl{PSPACE}}-complete.
}%
\label{theorem: KL extras PSPACE-complete}
\theoremKLextrasPSPACEcomplete
\end{theorem}

As an application,
we can also encode \kl{domain} and \kl{range} as above \eqref{equation: dom range by loop top}.
Consequently, our result also implies the \kl{PSPACE} upper bounds for 
\kl{RKAT} with \kl{domain} and \kl{RKA} with \kl{domain} \cite{desharnaisInternalAxiomsDomain2011,mcleanFreeKleeneAlgebras2020,sedlarKleeneAlgebraDynamic2023}\footnote{
\kl{RKA} with \kl{domain} is a strictly less expressive fragment of \kl{RKA} with \kl{tests} and \kl{antidomain}.
The latter system is often also called (relational) ``Kleene algebra with domain'' \cite{desharnaisKleeneAlgebraDomain2006}.},
which answers \cite[Problem 7.3]{mcleanFreeKleeneAlgebras2020} and \cite[p.\ 16]{sedlarKleeneAlgebraDynamic2023}:
\begin{cor}\label{corollary: KA with domain PSPACE-complete}
The \kl{equational theory}
for \kl{RKA} with \kl{tests}, \kl{domain}, and \kl{range} is \textup{\kl{PSPACE}}-complete.
\end{cor}
Notably, adding \kl{domain} to \kl{RKA} maintains the \kl{PSPACE} upper bound,
while the \kl{equational theory} of \kl{RKA} with \kl{antidomain} ($\bl^{\mathsf{a}}$) is \kl{EXPTIME}-complete \cite{sedlarComplexityKleeneAlgebra2023}.

To prove \Cref{theorem: KL PSPACE-complete,theorem: KL extras PSPACE-complete},
we give a reduction from the \kl{equational theory} into the \kl{language inclusion problem} of \kl{2-way alternating string automata} (\kl{2AFAs}) \cite{geffertTransformingTwoWayAlternating2014}.
We first observe that the \kl{equational theory} has the \emph{\kl{linearly bounded pathwidth model property}} (\Cref{proposition: bounded pw property}) \cite[Proposition 2.9]{nakamuraDerivativesGraphsPositive2025}:
if there is some counter-model for an \kl{equation} ("ie", a \kl{structure} refuting the \kl{equation}),
then there is also some counter-model with \kl{pathwidth} linear in the size of the \kl{equation}.
By viewing \kl{path decompositions} of bounded "width@@pathwidth" as \kl{strings},
we can enumerate all possible \kl{structures} of bounded \kl{pathwidth}, using \kl(2AFA){string automata}.
This baseline approach itself is the same as in \cite{nakamuraDerivativesGraphsPositive2025} for \kl{RKA} with \kl{intersection}.
However,
the construction of \cite{nakamuraDerivativesGraphsPositive2025} does not imply the \kl{PSPACE} upper bound (see also \cite[p.~46--47]{nakamuraDerivativesGraphsPositive2025}), especially if the \kl{intersection width} is not fixed.

To obtain a \kl{PSPACE}-algorithm,
we introduce a novel automaton model, called \emph{\kl{loop-automata}} (\Cref{section: loop-automata})
as an alternative expression of the relational semantics of \kl{loop-RKA}.
\kl[loop-automata]{Loop-automata} are obtained from nondeterministic automata by adding a new transition for testing whether there exists a loop in the current vertex.
From \kl{loop-automata}, we build \kl{2AFAs}, where each \kl{string} expresses a \kl{path decomposition}.
Using this construction,
we give a reduction from the \kl{equational theory} of \kl{loop-RKA} into the \kl{language inclusion problem} of \kl{2AFAs}.
In our automata construction, we can also encode several additional operators,
which shows \Cref{theorem: KL extras PSPACE-complete} and \Cref{corollary: KA with domain PSPACE-complete}.

\subparagraph*{Related and future work}
\AP
Our automata construction in this paper is based on \cite{nakamuraDerivativesGraphsPositive2025},
which shows the \kl{equational theory} for the \emph{positive} fragment of Tarski's calculus of relations with transitive closure (\intro*\kl{PCoR*}) \cite{tarskiCalculusRelations1941,pousPositiveCalculusRelations2018,nakamuraDerivativesGraphsPositive2025} is \kl{EXPSPACE}-complete.
While \kl{loop-RKA} can be viewed as a syntactic fragment of \kl{PCoR*},
the algorithm presented in \cite{nakamuraDerivativesGraphsPositive2025} requires an exponential space even for \kl{loop-RKA} \cite[p.~46--47]{nakamuraDerivativesGraphsPositive2025}.
Our automata construction in this paper refines it, specializing \kl{loop-RKA}.
Our reduction is different from the reduction of \cite{nakamuraDerivativesGraphsPositive2025}, mainly in the following two points:
\begin{itemize}
    \item
    We consider only (nested) \kl{path graphs} instead of \intro*\kl{series-parallel graphs},
    using \kl{loop-automata} instead of \intro*\kl(RKL){derivatives} on \kl{graphs}.
    These \kl{graphs} are sufficient for \kl{loop-RKA}.
    This significantly simplifies our algorithm compared to that of \cite{nakamuraDerivativesGraphsPositive2025}
    and enables us to obtain the \kl{PSPACE} upper bound.
    
    \item
    We reduce the alphabet size using a binary encoding of \kl{structures}.
    This is also crucial for our \kl{PSPACE} upper bound, especially if the \kl{intersection width} is not fixed.
\end{itemize}

Notably, \kl{RKAT} with \kl{domain} and \kl{range} has a connection with \intro*\kl{incorrectness logic} \cite{ohearnIncorrectnessLogic2019}.
An (angelic total) incorrectness triple ("ie", all states in $q$ can be reached by running $\term$ on some state in $p$)
can be expressed as follows, where $p, q$ are \kl{tests} and $\term$ is a \kl{KAT term}:
\begin{gather*}
[p]~\term~[q] \quad\leftrightarrow\quad (p \term q)^{\rng} \ge q \quad(\leftrightarrow\quad (p \term q)^{\rng} = q^{\rng} \quad\leftrightarrow\quad \top p \term q = \top q).
\end{gather*}
See also \cite[Table 2]{verschtTaxonomyHoareLikeLogics2025} for further correspondences between Hoare-like logics and \kl{RKAT} with \kl{domain} and \kl{range}, based on the equivalences $\term[1] \top = \term[2] \top \leftrightarrow \term[1]^{\dom} = \term[2]^{\dom}$ and $\top \term[1] = \top \term[2] \leftrightarrow \term[1]^{\rng} = \term[2]^{\rng}$.
In previous work, the \kl{PSPACE} upper bound of the validity of such incorrectness triples was shown via an encoding to \kl{KAT} with "top" over \kl{relational models} \cite{nakamuraExistentialCalculiRelations2023,pousCompletenessTheoremsKleene2024}, or \kl{KAT} with "top" over \kl{language models} \cite{zhangDomainReasoningTopKAT2024}.
Our \Cref{corollary: KA with domain PSPACE-complete} more directly shows that the above (in)"equation" is in \kl{PSPACE}.

In this paper, we present an algorithm for \kl{loop-RKA}; however,
it would also be interesting to present an algorithm specialized to \kl{RKA} with \kl{domain}.\footnote{%
Our approach of encoding \kl{domain} using \kl{graph loop} and \kl{top} does not preserve the complexity in general.
For example, whereas the "equational theory" of \kl{RKA} with \kl{domain} and \kl{variable complements}
is in \kl{EXPTIME} (by \cite{lutzPDLNegationAtomic2005}), it is undecidable for \kl{loop-RKA} with \kl{variable complements} \cite{nakamuraExistentialCalculiRelations2023,nakamuraUndecidabilityEmptinessProblem2025} (\Cref{footnote: undecidable}).}

Additionally, it is open whether the \kl{equational theory} of \kl{loop-RKA} ("resp", ""loop-RKAT"", "ie",
"RKAT" with "graph loop") is finitely axiomatizable.
An interesting valid "equation" in \kl{loop-RKA} is $(\term^{+})^{\lop} \le (\term \term)^{+}$ (cf.~\cite[p.~14]{pousPositiveCalculusRelations2018}).

\subparagraph*{Organization}
In \Cref{section: preliminaries}, we give basic definitions.
In \Cref{section: RKA with graph loop}, we define \kl{loop-RKA}.
In \Cref{section: loop-automata}, we introduce \kl{loop-automata} as an alternative representation of \kl{loop-RKA terms}.
In \Cref{section: decomposing derivatives,section: automata construction},
we show that the \kl{equational theory} of \kl{loop-RKA} is \kl{PSPACE}-complete (\Cref{theorem: KL PSPACE-complete}).
In \Cref{section: encoding extras}, we extend the automata construction with several additional operators (\Cref{theorem: KL extras PSPACE-complete} and \Cref{corollary: KA with domain PSPACE-complete}).
\end{scope}

\section{Preliminaries}\label{section: preliminaries}
We write $\nat$ for the set of non-negative integers.
\AP
For $l, r \in \nat$, we write $\intro*\range{l}{r}$ for the set $\set{i \in \nat \mid l \le i \le r}$.
\AP
For $n \in \nat$, we abbreviate $\range{1}{n}$ to $\intro*\rangeone{n}$.
\AP
For a set $X$,
we write $\intro*\card X$ for the cardinality of $X$. %
\AP
We may use $\intro*\dcup$ to denote that the union $\cup$ is disjoint, explicitly.
\AP For a sequence $\mul{a} = a_1 \dots a_n$,
we write $\mul{a}\intro*\getith{i}$ for the element $a_i$.

\subsection{Word languages}
\AP
For a set $X$ of \intro*\kl{characters}, we write $X^{*}$ for the set of \intro*\kl{strings} over $X$.
\AP
We write $\word[1] \word[2]$ for the \intro*\kl(word){concatenation} of \kl{strings} $\word[1]$ and $\word[2]$.
\AP
We write $\intro*\eps$ for the \intro*\kl{empty string}.
\AP
A \intro*\kl{language} over $X$ is a subset of $X^{*}$.
\AP
We use $\word[1], \word[2]$ to denote \kl{strings} and use $\la[1], \la[2]$ to denote \kl{languages}.
\AP
For \kl{languages} $\la[1], \la[2] \subseteq X^{*}$, the \intro*\kl(lang){concatenation} $\la[1] \compo \la[2]$,
the \intro*\kl(lang){$n$-th iteration} $\la[1]^n$ (where $n \in \nat$),
the \intro*\kl{Kleene plus} $\la[1]^{+}$, and
the \intro*\kl{Kleene star} $\la[1]^{*}$ are defined by:
\begin{align*}
    \la[1] \compo \la[2] & \defeq \set{\word[1] \word[2] \mid \word[1] \in \la[1] \;\land\; \word[2] \in \la[2]},\hspace{-.1em}                    \span\span    \\
    \la[1]^{n} & \defeq \begin{cases}
        \la[1] \compo \la[1]^{n-1} & (n \ge 1)\\
        \set{\eps} & (n = 0),
    \end{cases}\hspace{-.1em} & 
    \la[1]^{+}  & \defeq \bigcup_{n \ge 1} \la[1]^{n},\hspace{-.1em}&
    \la[1]^{*}  & \defeq \bigcup_{n \ge 0} \la[1]^{n}.
\end{align*}

\subsection{Binary relations}
\AP
We write $\intro*\diagonal_{A}$ for the \intro*\kl{identity relation} on a set $A$: $\diagonal_{A} \defeq \set{\tuple{x, x} \mid x \in A}$.
\AP
For \intro*\kl{binary relations} $\rel[1], \rel[2]$ on a set $B$, the \intro*\kl{composition} $\rel[1] \compo \rel[2]$, the \AP
    \intro*\kl{$n$-th iteration} $\rel[1]^{n}$ (where $n \in \nat$), the \intro*\kl{transitive closure} $\rel[1]^{+}$, the \intro*\kl{reflexive transitive closure} $\rel[1]^*$,
    the \intro*\kl{domain} $\rel[1]^{\dom}$,
the \intro*\kl{range} $\rel[1]^{\rng}$, and 
the \intro*\kl{graph loop} $\rel[1]^{\lop}$
are defined by:
\begin{align*}
    \rel[1] \compo \rel[2] & \defeq \set{\tuple{x, z} \mid \exists y, \tuple{x, y} \in \rel[1] \;\land\;  \tuple{y, z} \in \rel[2]}, &
    \rel^{\lop} &\defeq \rel \cap \diagonal_{B}, \\
    \rel[1]^n             & \defeq \begin{cases}
                                       \rel[1] \compo \rel[1]^{n-1} & (n \ge 1) \\
                                       \diagonal_{B}               & (n = 0)
                                   \end{cases}, &
    \rel[1]^{+}        & \defeq \bigcup_{n \ge 1} \rel[1]^n, \quad
    \rel[1]^*  \defeq \bigcup_{n \ge 0} \rel[1]^n,\\
    \rel^{\dom} &\defeq \set{\tuple{x, x} \mid \exists y, \tuple{x, y} \in \rel}, & 
    \rel^{\rng} &\defeq \set{\tuple{y, y} \mid \exists x, \tuple{x, y} \in \rel}. 
\end{align*}

\subsection{Graphs and structures}
\AP
For $k \ge 0$ and a set $A$,
a \intro*\kl{graph} $\graph$ over $A$ with $k$ ""ports"" ("ie", $k$ named vertices) is a tuple $\tuple{\univ{\graph}, \set{a^{\graph}}_{a \in A}, 1^{\graph}, \dots, k^{\graph}}$, where
\begin{itemize} 
    \item \AP$\intro*\univ{\graph}$ is a set (which is possibly empty when $k = 0$) of \intro*\kl{vertices},
    \item \AP$a^{\graph} \subseteq \univ{\graph}^{2}$ is a binary relation for each $a \in A$ denoting the set of \intro*\kl{edges} labelled by $a$,
    \item $i^{\graph} \in \univ{\graph}$ is the $i$-th port for each $i \in \rangeone{k}$.
\end{itemize}
We depict \kl{graphs} as usual.
For instance, when $\graph$ is defined as $\univ{\graph} = \set{1, 2, 3}$, $a^{\graph} = \set{\tuple{1, 2}, \tuple{2, 3}}$, $1^{\graph} = 1$, and $2^{\graph} = 3$,
we depict $\graph$ as
\begin{tikzpicture}[baseline = -.5ex]
    \graph[grow right = 1.cm, branch down = 6ex, nodes={mynode}]{
    {1/{$1$}[draw, circle]}-!-{2/{$2$}[draw, circle]}-!-{3/{$3$}[draw, circle]}
    };
    \node[left = .5em of 1](l){\scriptsize $1$};
    \node[right = .5em of 3](r){\scriptsize $2$};
    \graph[use existing nodes, edges={color=black, pos = .5, earrow}, edge quotes={fill=white, inner sep=1pt,font= \scriptsize}]{
    1 ->["$a$"] 2 ->["$a$"] 3;
    l -- 1; 3 -- r;
    };
\end{tikzpicture}.
By viewing the "ports" $1^{\graph}$ and $2^{\graph}$ as the ``input'' and ``output'', respectively,
we may depict $\graph$ as \begin{tikzpicture}[baseline = -.5ex]
    \graph[grow right = 1.cm, branch down = 6ex, nodes={mynode}]{
    {1/{}[draw, circle]}-!-{2/{}[draw, circle]}-!-{3/{}[draw, circle]}
    };
    \node[left = .5em of 1](l){};
    \node[right = .5em of 3](r){};
    \graph[use existing nodes, edges={color=black, pos = .5, earrow}, edge quotes={fill=white, inner sep=1pt,font= \scriptsize}]{
    1 ->["$a$"] 2 ->["$a$"] 3;
    l -> 1; 3 -> r;
    };
\end{tikzpicture} (by drawing \begin{tikzpicture}[baseline = -.5ex]
\graph[grow right = 1.cm, branch down = 6ex, nodes={mynode}]{
{1/{}[draw, circle]}
};
\node[left = .5em of 1](l){\scriptsize $1$};
\graph[use existing nodes, edges={color=black, pos = .5, earrow}, edge quotes={fill=white, inner sep=1pt,font= \scriptsize}]{
    l -- 1;
};
\end{tikzpicture} as \begin{tikzpicture}[baseline = -.5ex]
\graph[grow right = 1.cm, branch down = 6ex, nodes={mynode}]{
{1/{}[draw, circle]}
};
\node[left = .5em of 1](l){};
\graph[use existing nodes, edges={color=black, pos = .5, earrow}, edge quotes={fill=white, inner sep=1pt,font= \scriptsize}]{
    l -> 1;
};
\end{tikzpicture} and
\begin{tikzpicture}[baseline = -.5ex]
\graph[grow right = 1.cm, branch down = 6ex, nodes={mynode}]{
{1/{}[draw, circle]}
};
\node[right = .5em of 1](r){\scriptsize $2$};
\graph[use existing nodes, edges={color=black, pos = .5, earrow}, edge quotes={fill=white, inner sep=1pt,font= \scriptsize}]{
    1 -- r;
};
\end{tikzpicture} as \begin{tikzpicture}[baseline = -.5ex]
\graph[grow right = 1.cm, branch down = 6ex, nodes={mynode}]{
{1/{}[draw, circle]}
};
\node[right = .5em of 1](r){};
\graph[use existing nodes, edges={color=black, pos = .5, earrow}, edge quotes={fill=white, inner sep=1pt,font= \scriptsize}]{
    1 -> r;
};
\end{tikzpicture}, respectively, and by forgetting vertex labels), \begin{tikzpicture}[baseline = -.5ex]
    \graph[grow right = 1.cm, branch down = 6ex, nodes={mynode}]{
    {1/{}[draw, circle]}-!-{2/{}[draw, circle]}-!-{3/{}[draw, circle]}
    };
    \node[left = .5em of 1](l){};
    \node[right = .5em of 3](r){};
    \graph[use existing nodes, edges={color=black, pos = .5, earrow}, edge quotes={fill=white, inner sep=1pt,font= \scriptsize}]{
    1 -> 2 -> 3;
    l -> 1; 3 -> r;
    };
\end{tikzpicture} (by forgetting edge labels $a$), or \begin{tikzpicture}[baseline = -.5ex]
    \graph[grow right = 1.5cm, branch down = 6ex, nodes={mynode}]{
    {1/{}[draw,circle]}-!-{3/{}[draw,circle]}
    };
    \graph[use existing nodes, edges={color=black, pos = .5, earrow}, edge quotes={fill=white, inner sep=1pt,font= \scriptsize}]{
    1 ->["$\graph$"{draw,rectangle}] 3;
    };
\end{tikzpicture}.
For two \kl{graphs} $\graph[1]$ and $\graph[2]$ over $A$ with the same number of ports,
a \intro*\kl{graph isomorphism} is a bijection $f \colon \univ{\graph[1]} \to \univ{\graph[2]}$
that preserves and reflects every relation $a \in A$ and maps each port of $\graph[1]$ to the corresponding port of $\graph[2]$.
We write $\graph[1] \cong \graph[2]$ if such an isomorphism exists.

\AP
A \intro*\kl{structure} over $A$ is a non-empty graph without ports over $A$.
We use $\struc[1], \struc[2]$ to denote \kl{structures}.
We write $\intro*\REL$ for the class of all \kl{structures}.
Whenever a signature is fixed, such as $\vsig$ below, $\REL$ denotes the structures over that signature.

\subsection{Pathwidth, path decomposition, and gluing operator}
\AP
We recall the \kl{pathwidth} of \kl{structures} \cite[Def.\ 9.12]{courcelleGraphStructureMonadic2012}.
For a finite \kl{structure} $\struc[1]$ over $A$,
a \intro*\kl{path decomposition} of $\struc[1]$ is a sequence $\mul{\struc[2]} = \struc[2]_1 \dots \struc[2]_n$ of finite \kl{structures} such that
\begin{itemize}
    \item $\univ{\struc[1]} = \bigcup_{i \in \rangeone{n}} \univ{\struc[2]_{i}}$ and $a^{\struc[1]} = \bigcup_{i \in \rangeone{n}} a^{\struc[2]_{i}}$ for each $a \in A$,
    \item $\univ{\struc[2]_i} \cap \univ{\struc[2]_k} \subseteq \univ{\struc[2]_j} \mbox{ for all $1 \le i \le j \le k \le n$}$.
\end{itemize}
Each $\struc[2]_i$ is called a \intro*\kl{bag}.
\AP
The \intro*\kl(pathwidth){width} of $\mul{\struc[2]}$ is $\max_{i \in \rangeone{n}}(\card \univ{\struc[2]_i} - 1)$.
The \intro*\kl{pathwidth} $\intro*\pw(\struc[1])$ of a finite \kl{structure} $\struc[1]$ is defined as the minimum \kl(pathwidth){width} among \kl{path decompositions} of $\struc[1]$.
For instance, for the $\struc$ and $\mul{\struc[2]}$ in \Cref{figure: glue},
$\mul{\struc[2]}$ is a \kl{path decomposition} of $\struc$ of \kl(pathwidth){width} $2$.
\begin{figure}[h]
    \begin{align*}
    \struc =& \left(\begin{tikzpicture}[baseline = -3.5ex]
    \graph[grow right = .9cm, branch down = 3ex, nodes={}]{
    {1/{$\mathsf{1}$}[mynode, draw, circle]}-!-{/, 2/{$\mathsf{2}$}[mynode, draw, circle]}-!-{/, 3/{$\mathsf{3}$}[mynode, draw, circle], 4/{$\mathsf{4}$}[mynode, draw, circle]}-!-
    {5/{$\mathsf{5}$}[mynode, draw, circle]}-!-{/, 6/{$\mathsf{6}$}[mynode, draw, circle]}-!-{/, 7/{$\mathsf{7}$}[mynode, draw, circle], 8/{$\mathsf{8}$}[mynode, draw, circle]}-!-
    {9/{$\mathsf{9}$}[mynode, draw, circle]}-!-{/, a/{$\mathsf{a}$}[mynode, draw, circle]}-!-{/, b/{$\mathsf{b}$}[mynode, draw, circle], c/{$\mathsf{c}$}[mynode, draw, circle]}-!-
    {d/{$\mathsf{d}$}[mynode, draw, circle]}
    };
    \graph[use existing nodes, edges={color=black, pos = .5, earrow}, edge quotes={fill=white, inner sep=1pt,font= \scriptsize}]{
    1 -> 2 -> {3, 4};
    5 -> 6 -> {7, 8};
    9 -> a -> {b, c};
    1 -> 5 -> 9 -> d;
    };
\end{tikzpicture}\right), \\[1ex]
    \mul{\struc[2]} =& \left(\begin{tikzpicture}[baseline = -2.ex, remember picture]
        \graph[grow right = .7cm, branch down = 3ex, nodes={}]{
        {11/{$\mathsf{1}$}[mynode, draw, circle]}-!-{/, 12/{$\mathsf{2}$}[mynode, draw, circle]}-!-{/, 13/{$\mathsf{3}$}[mynode, draw, circle]}
        };
        \graph[use existing nodes, edges={color=black, pos = .5, earrow}, edge quotes={fill=white, inner sep=1pt,font= \scriptsize}]{
        11 -> 12 -> 13;
        };
    \end{tikzpicture}\right)
    \left(\begin{tikzpicture}[baseline = -2.ex, remember picture]
        \graph[grow right = .7cm, branch down = 3ex, nodes={}]{
        {21/{$\mathsf{1}$}[mynode, draw, circle]}-!-{/, 22/{$\mathsf{2}$}[mynode, draw, circle]}-!-{/, 24/{$\mathsf{4}$}[mynode, draw, circle]}
        };
        \graph[use existing nodes, edges={color=black, pos = .5, earrow}, edge quotes={fill=white, inner sep=1pt,font= \scriptsize}]{
        21 -> 22 -> 24;
        };
    \end{tikzpicture}\right)
    \left(\begin{tikzpicture}[baseline = -2.ex, remember picture]
        \graph[grow right = .7cm, branch down = 3ex, nodes={}]{
        {31/{$\mathsf{1}$}[mynode, draw, circle]}-!-{/}-!-
        {35/{$\mathsf{5}$}[mynode, draw, circle]}
        };
        \graph[use existing nodes, edges={color=black, pos = .5, earrow}, edge quotes={fill=white, inner sep=1pt,font= \scriptsize}]{
        31 -> 35;
        };
    \end{tikzpicture}\right)
    \left(\begin{tikzpicture}[baseline = -2.ex, remember picture]
        \graph[grow right = .7cm, branch down = 3ex, nodes={}]{
        {45/{$\mathsf{5}$}[mynode, draw, circle]}-!-{/, 46/{$\mathsf{6}$}[mynode, draw, circle]}-!-{/, 47/{$\mathsf{7}$}[mynode, draw, circle]}
        };
        \graph[use existing nodes, edges={color=black, pos = .5, earrow}, edge quotes={fill=white, inner sep=1pt,font= \scriptsize}]{
        45 -> 46 -> 47;
        };
    \end{tikzpicture}\right)
    \left(\begin{tikzpicture}[baseline = -2.ex, remember picture]
        \graph[grow right = .7cm, branch down = 3ex, nodes={}]{
        {55/{$\mathsf{5}$}[mynode, draw, circle]}-!-{/, 56/{$\mathsf{6}$}[mynode, draw, circle]}-!-{/, 58/{$\mathsf{8}$}[mynode, draw, circle]}
        };
        \graph[use existing nodes, edges={color=black, pos = .5, earrow}, edge quotes={fill=white, inner sep=1pt,font= \scriptsize}]{
        55 -> 56 -> 58;
        };
    \end{tikzpicture}\right)\\[3ex]
    &\left(\begin{tikzpicture}[baseline = -2.ex, remember picture]
        \graph[grow right = .7cm, branch down = 3ex, nodes={}]{
        {65/{$\mathsf{5}$}[mynode, draw, circle]}-!-{/}-!-
        {69/{$\mathsf{9}$}[mynode, draw, circle]}
        };
        \graph[use existing nodes, edges={color=black, pos = .5, earrow}, edge quotes={fill=white, inner sep=1pt,font= \scriptsize}]{
        65 -> 69;
        };
    \end{tikzpicture}\right)
    \left(\begin{tikzpicture}[baseline = -2.ex, remember picture]
        \graph[grow right = .7cm, branch down = 3ex, nodes={}]{
        {79/{$\mathsf{9}$}[mynode, draw, circle]}-!-{/, 7a/{$\mathsf{a}$}[mynode, draw, circle]}-!-{/, 7b/{$\mathsf{b}$}[mynode, draw, circle]}
        };
        \graph[use existing nodes, edges={color=black, pos = .5, earrow}, edge quotes={fill=white, inner sep=1pt,font= \scriptsize}]{
        79 -> 7a -> {7b};
        };
    \end{tikzpicture}\right)
    \left(\begin{tikzpicture}[baseline = -2.ex, remember picture]
        \graph[grow right = .7cm, branch down = 3ex, nodes={}]{
        {89/{$\mathsf{9}$}[mynode, draw, circle]}-!-{/, 8a/{$\mathsf{a}$}[mynode, draw, circle]}-!-{/, 8c/{$\mathsf{c}$}[mynode, draw, circle]}
        };
        \graph[use existing nodes, edges={color=black, pos = .5, earrow}, edge quotes={fill=white, inner sep=1pt,font= \scriptsize}]{
        89 -> 8a -> {8c};
        };
    \end{tikzpicture}\right)
    \left(\begin{tikzpicture}[baseline = -2.ex, remember picture]
        \graph[grow right = .7cm, branch down = 3ex, nodes={}]{
        {99/{$\mathsf{9}$}[mynode, draw, circle]}-!-{/}-!-
        {9d/{$\mathsf{d}$}[mynode, draw, circle]}
        };
        \graph[use existing nodes, edges={color=black, pos = .5, earrow}, edge quotes={fill=white, inner sep=1pt,font= \scriptsize}]{
        99 -> 9d;
        };
    \end{tikzpicture}\right),\\[1ex]
     &\quad\leadsto\quad \glue \mul{\struc[2]} = \left(\begin{tikzpicture}[baseline = -3.5ex]
    \graph[grow right = .7cm, branch down = 2.5ex, nodes={inner sep = 6pt, minimum size = 6pt}]{
    {1/{}[mynode, draw, circle]}-!-{/, 2/{}[mynode, draw, circle]}-!-{/, 3/{}[mynode, draw, circle], 4/{}[mynode, draw, circle]}-!-
    {5/{}[mynode, draw, circle]}-!-{/, 6/{}[mynode, draw, circle]}-!-{/, 7/{}[mynode, draw, circle], 8/{}[mynode, draw, circle]}-!-
    {9/{}[mynode, draw, circle]}-!-{/, a/{}[mynode, draw, circle]}-!-{/, b/{}[mynode, draw, circle], c/{}[mynode, draw, circle]}-!-
    {d/{}[mynode, draw, circle]}
    };
    \graph[use existing nodes, edges={color=black, pos = .5, earrow}, edge quotes={fill=white, inner sep=1pt,font= \scriptsize}]{
    1 -> 2 -> {3, 4};
    5 -> 6 -> {7, 8};
    9 -> a -> {b, c};
    1 -> 5 -> 9 -> d;
    };
\end{tikzpicture}\right).
    \end{align*}
    \begin{tikzpicture}[remember picture, overlay]
    \node[below = 5ex of 55.south, inner sep = 0, minimum size = 0](55b){};
    \node[above = 2ex of 65.north, inner sep = 0, minimum size = 0](65a){};
    \graph[use existing nodes, edges={color=olive, pos = .5, earrow, densely dashed, line width = 1pt,}, edge quotes={draw = gray, inner sep=1pt,font= \scriptsize}]{
                11 --[out = 45, in = 135, looseness = .5] 21 --[out = 45, in = 135, looseness = .5] 31;
                12 --[out = -45, in = -135, looseness = .5] 22;
                35 --[out = 45, in = 135, looseness = 1.] 45 --[out = 45, in = 135, looseness = 0.5] 55;
                55.south -- 55b -- 65a --[out = -135, in = 45, looseness = 0] 65.north;
                46 --[out = -45, in = -135, looseness = .5] 56;
                69 --[out = 45, in = 135, looseness = 1.] 79 --[out = 45, in = 135, looseness = .5] 89 --[out = 45, in = 135, looseness = .5] 99;
                7a --[out = -45, in = -135, looseness = .5] 8a;
            };
    \end{tikzpicture}     \caption{Illustrative example of \kl{path decomposition} and the gluing operator $\glue$.
    $\mul{\struc[2]}$ is a \kl{path decomposition} of $\struc[1]$.
    $\glue \mul{\struc[2]}$ is obtained from $\mul{\struc[2]}$ by gluing adjacent \kl{vertices} with the same label (indicated by dotted lines).}
    \label{figure: glue}
\end{figure}

\phantomintro{gluing operator}
In this paper, we will view \kl{path decompositions} as \emph{\kl{strings}}, based on, "eg", \cite{nakamuraPartialDerivativesGraphs2017,nakamuraDerivativesGraphsPositive2025}.
\AP \phantomintro{\bigglue}
We define $\intro*\glue \mul{\struc}$ \cite[Definition 5.1]{nakamuraDerivativesGraphsPositive2025} as the \kl{structure}
obtained from the \kl{disjoint union} of \kl{structures} in $\mul{\struc}$ by gluing \kl{vertices} having the same name in adjacent \kl{structures}.
We write $\intro*\quo{\tuple{i, x}}_{\sim}$ for the 
equivalence class of the vertex named by $x$ in $\mul{\struc}\getith{i}$ "wrt" the equivalence relation induced from the construction of $\glue \mul{\struc}$; more precisely, it is the minimal equivalence relation $\sim$
satisfying $\tuple{i, v} \sim \tuple{i+1, v}$ for all $i \in \range{1}{n-1}$ and $v \in \univ{\mul{\struc}\getith{i}} \cap \univ{\mul{\struc}\getith{i+1}}$.
More explicitly, its universe and relations are
\begin{align*}
\univ{\glue\mul{\struc}}
&= \set{\quo{\tuple{i,x}}_{\sim} \mid i \in \rangeone{n},\ x \in \univ{\struc_i}},\\
a^{\glue\mul{\struc}}
&= \set{\tuple{\quo{\tuple{i,x}}_{\sim},\quo{\tuple{i,y}}_{\sim}} \mid i \in \rangeone{n},\ \tuple{x,y} \in a^{\struc_i}}
\quad (a \in A).
\end{align*}
\begin{example}
For the $\mul{\struc[2]}$ and $\struc[1]$ in \Cref{figure: glue},
the "structure" $\glue \mul{\struc[2]}$ is illustrated as in \Cref{figure: glue} by gluing "vertices" according to the dotted lines.
We observe that the "structure" is isomorphic to $\struc[1]$.
In general, the \kl{gluing operator} transforms a given \kl{path decomposition} into the original \kl{structure} up to \kl[graph isomorphism]{isomorphisms}.
\end{example}
\begin{example}
    Below is another example of the \kl{gluing operator} $\glue$:
\[\bigglue \left(\begin{tikzpicture}[baseline = -3.5ex, remember picture]
    \graph[grow right = 1.cm, branch down = 5.5ex]{
    {1s1/{$1$}[vert]} -!- {1t1/{$2$}[vert], 1s4/{$4$}[vert, xshift = -.5cm]}
    };
    \graph[use existing nodes, edges={color=black, pos = .5, earrow}, edge quotes={fill=white, inner sep=1pt,font= \scriptsize}]{
        1s1 ->["$a$"] 1t1;
        1s4 ->["$b$"] 1s1;
    };
\end{tikzpicture} \right) \left(\begin{tikzpicture}[baseline = -3.5ex, remember picture]
    \graph[grow right = 1.cm, branch down = 5.5ex]{
    {2s1/{$2$}[vert]} -!- {2t1/{$3$}[vert], 2s4/{$4$}[vert, xshift = -.5cm]}
    };
    \graph[use existing nodes, edges={color=black, pos = .5, earrow}, edge quotes={fill=white, inner sep=1pt,font= \scriptsize}]{
        2s1 ->["$a$"] 2t1;
    };
\end{tikzpicture} \right) \left(\begin{tikzpicture}[baseline = -3.5ex, remember picture]
    \graph[grow right = 1.cm, branch down = 5.5ex]{
    {3s1/{$3$}[vert]} -!- {3t1/{$1$}[vert], 3s4/{$4$}[vert, xshift = -.5cm]}
    };
    \graph[use existing nodes, edges={color=black, pos = .5, earrow}, edge quotes={fill=white, inner sep=1pt,font= \scriptsize}]{
        3s1 ->["$a$"] 3t1 ->["$b$"] 3s4;
        3s4 ->["$a$"] 3s1;
    };
\end{tikzpicture} \right) = \left(\begin{tikzpicture}[baseline = -4.ex]
    \graph[grow right = 1.7cm, branch down = 7.5ex]{
    {1/{$[\tuple{1, 1}]_{\sim}$}[vert, font = \tiny]} -!-
    {2/{$[\tuple{1, 2}]_{\sim}$}[vert, font = \tiny]} -!-
    {3/{$[\tuple{2, 3}]_{\sim}$}[vert, font = \tiny], 5/{$[\tuple{3, 4}]_{\sim}$}[vert, font = \tiny, xshift = -1.cm]} -!-
    {4/{$[\tuple{3, 1}]_{\sim}$}[vert, font = \tiny]}
    };
    \graph[use existing nodes, edges={color=black, pos = .5, earrow}, edge quotes={fill=white, inner sep=1pt,font= \scriptsize}]{
        1 ->["$a$"] 2 ->["$a$"] 3 ->["$a$"] 4;
        4 -> ["$b$", out = -135, in = 0, looseness = 0] 5 ->["$b$", out = 180, in = -45, looseness = 0] 1;
        5 ->["$a$", out = 90, in = -135, looseness = 0.5] 3; 
    };
    \end{tikzpicture}
    \right).%
    \begin{tikzpicture}[remember picture, overlay]
    \graph[use existing nodes, edges={color=olive, pos = .5, earrow, densely dashed, line width = 1pt,}, edge quotes={draw = gray, inner sep=1pt,font= \scriptsize}]{
        1t1 --[out = 60, in = 120, looseness = 1.] 2s1;
        2t1 --[out = 60, in = 120, looseness = 1.] 3s1;
        1s4 --[out = -60, in = -120, looseness = .5] 2s4  --[out = -60, in = -120, looseness = .5] 3s4;
    };
    \end{tikzpicture}\]
    Note that the "vertices" labelled with $1$ are not glued since they are not adjacent.
\end{example}

\section{Relational Kleene Algebra with Graph Loop}\label{section: RKA with graph loop}
\AP
In this section,
we define \intro*\kl{relational Kleene algebra with graph loop} (\reintro*\kl{loop-RKA}).

\subparagraph*{Syntax}
\AP
Given a set $\vsig$ of \intro*\kl{variables},
the \intro*\kl(loop-RKA){terms} over $\vsig$ 
are defined as the terms over the signature $\set{\emp_{(0)}, \id_{(0)}, \union_{(2)}, \compo_{(2)}, \bl^{*}_{(1)}, \bl^{\lop}_{(1)}}$: 
\begin{align*}
    \term[1], \term[2], \term[3] \quad\Coloneqq\quad \aterm \mid \emp \mid \id \mid \term[1] \union \term[2] \mid \term[1] \compo \term[2] \mid \term[1]^* \mid \term[1]^{\lop}. \tag*{where $a \in \vsig$}
\end{align*}
We often abbreviate $\term[1] \compo \term[2]$ to $\term[1] \term[2]$.
We use parentheses in ambiguous situations.
We write $\sum_{i = 1}^{n} \term[1]_i$ for the term $\emp \union \term[1]_1 \union \dots \union \term[1]_n$.
We also let $\term[1]^{+} \defeq \term[1] \term[1]^*$.
\AP
An \intro*\kl{equation} $\term[1] = \term[2]$ is a pair of \kl(loop-RKA){terms}.
An \intro*\kl{inequation} $\term[1] \le \term[2]$ abbreviates the \kl{equation} $\term[1] \union \term[2] = \term[2]$.
The \intro*\kl{size} $\|\term\|$ of a \kl(loop-RKA){term} $\term$ is defined as the number of symbols occurring in $\term$.

\subparagraph*{Semantics}
\AP
For a \kl{structure} $\struc[1]$ over $\vsig$
and a \kl(loop-RKA){term} $\term$,
the \intro*\kl{semantics}\footnote{
For notational convenience, we use \kl{structures} as an alternative to algebras of binary relations.}
$\intro*\jump{\term}^{\struc} \subseteq \univ{\struc}^2$ is defined as follows:
\begin{gather*}
    \jump{a}^{\struc} \defeq a^{\struc}, \qquad
    \jump{\emp}^{\struc} \defeq \emptyset, \qquad
    \jump{\id}^{\struc} \defeq \diagonal_{\univ{\struc}}, \qquad
    \jump{\term[1] \union \term[2]}^{\struc} \defeq \jump{\term[1]}^{\struc} \cup \jump{\term[2]}^{\struc}, \\
    \jump{\term[1] \compo \term[2]}^{\struc} \defeq \jump{\term[1]}^{\struc} \compo \jump{\term[2]}^{\struc}, \qquad
    \jump{\term[1]^*}^{\struc} \defeq (\jump{\term[1]}^{\struc})^*, \qquad
    \jump{\term[1]^{\lop}}^{\struc} \defeq \jump{\term[1]}^{\struc} \cap \diagonal_{\univ{\struc}}.
\end{gather*}
Recall that $\REL$ is the class of all \kl{structures}.
\AP
For $\algclass \subseteq \REL$,
we write $\algclass \intro*\klmodels \term[1] = \term[2]$ if $\jump{\term[1]}^{\struc} = \jump{\term[2]}^{\struc}$ for all $\struc \in \algclass$.
In the sequel, we mainly consider the \intro*\kl{equational theory} "wrt" $\REL$.

\subparagraph*{Linearly bounded pathwidth model property}
\AP
The \intro*\kl{intersection width} $\intro*\iw(\term)$ \cite{gollerPDLIntersectionConverse2009} (where the \kl(loop-RKA){term} $\term^{\lop}$ is viewed as $\term \cap \id$) is defined as follows:
\begin{align*}
    \iw(\aterm) &\defeq 1  \mbox{ for $\aterm \in \vsig \cup \set{\id, \emp}$}, & 
    \iw(\term[1]^{*}) &\defeq \iw(\term[1]),\\
    \iw(\term[1] \mathbin{\heartsuit} \term[2]) &\defeq \max(\iw(\term[1]), \iw(\term[2])) \mbox{  for $\mathbin{\heartsuit} \in \set{\compo, \union}$}, &
    \iw(\term[1]^{\lop}) &\defeq \iw(\term[1]) + 1.
\end{align*}
For instance, $\iw(a^{\lop} b^{\lop} c^{\lop}) = 2$ and $\iw((c (b a^{\lop} b)^{\lop} c)^{\lop}) = 4$.
\begin{proposition}\label{proposition: iw}
    For all \kl(loop-RKA){terms} $\term$, we have $\iw(\term) \le \|\term\|$.
\end{proposition}
\begin{proof}
    By easy induction on $\term$.
\end{proof}

For a class $\mathcal{C} \subseteq \REL$, we let
$\mathcal{C}_{\pw \le k} = \set{\struc \in \mathcal{C} \mid \struc \text{ is finite and } \pw(\struc) \le k}$.
Because \kl{loop-RKA} can be viewed as a syntactic fragment of "PCoR*" \cite{pousPositiveCalculusRelations2018,nakamuraDerivativesGraphsPositive2025},
\AP
we have the \intro*\kl{linearly bounded pathwidth model property}:
\begin{proposition}[{\cite[Proposition 2.9]{nakamuraDerivativesGraphsPositive2025}}]\label{proposition: bounded pw property}
    For all \kl{loop-RKA terms}, $\term[1]$ and $\term[2]$,
    we have:
    \[\REL \klmodels \term[1] \le \term[2] \quad\iff\quad \REL_{\pw \le \iw(\term[1])} \klmodels \term[1] \le \term[2].\]
\end{proposition}
For instance, when $\term[1] = ((a (b^+)^{\lop} a)^{\lop} c)^+$ and $\term[2] = (a (b + bb) a c^*) \union (c a b a c^*)$,
to show $\REL \not\klmodels \term[1] \le \term[2]$,
it suffices to consider \kl{structures} in $\REL_{\pw \le 3}$ by $\iw(\term[1]) = 3$.
For instance,
the following \kl{structure} $\struc \in \REL_{\pw \le 3}$
satisfies $\tuple{x, y} \in \jump{\term[1]}^{\struc} \setminus \jump{\term[2]}^{\struc}$ (hence, $\REL \not\klmodels \term[1] \le \term[2]$):
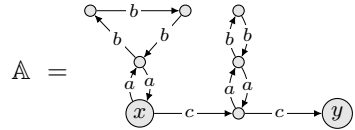
\begin{center}
    $\struc ~=~ \begin{tikzpicture}[baseline = -6.0ex]
        \graph[grow right = 1.3cm, branch down = 4.5ex]{
        {12/{}[vert, xshift = .65cm]} -!- {13/{}[vert, xshift = .6cm]/,11/{}[vert], 1/{$x$}[vert]} -!- {22/{}[vert], 21/{}[vert], 2/{}[vert]} -!- {/, /, 3/{$y$}[vert]}
        };
        \graph[use existing nodes, edges={color=black, pos = .45, earrow}, edge quotes={fill=white, inner sep=1pt,font= \scriptsize}]{
            1 ->["$a$", bend left] 11 ->["$a$", bend left] 1;
            11 ->["$b$"] 12 ->["$b$"] 13 ->["$b$"] 11;
            2 ->["$a$", bend left] 21 ->["$a$", bend left] 2;
            21 ->["$b$", bend left] 22 ->["$b$", bend left] 21;
            1 ->["$c$"]  2 ->["$c$"] 3;
        };
    \end{tikzpicture}$
\end{center}

\section{Automata Expression for Loop-RKA}\label{section: loop-automata}
\AP
In this section,
we introduce an automaton model, named \emph{\kl{loop-automata}},
for representing the \kl{relational semantics} of \kl{loop-RKA}.
\kl[loop-automata]{Loop-automata} accept/reject ordered pairs of "vertices" in "graphs" ("structures").
Given two finite sets $\Sigma$ (of \kl{characters}) and $Q$ (of \intro*\kl{states}),
we consider the following set of transition labels:
\AP
\[\intro*\LoopLabel{\Sigma, Q} ~\defeq~ \Sigma \dcup \set{\id} \dcup \set{\ell_{\tuple{p, q}} \mid \tuple{p, q} \in Q^2}.\]
Each label $\aterm \in \Sigma$
is employed for the standard transitions.
The label $\id$ is employed for the ``epsilon''-transitions:
transitions that do not consume any input symbol.
Each label $\ell_{\tuple{p, q}}$ is employed for expressing \emph{loop-transitions}, which are conditional ``epsilon''-transitions: the transition is allowed only when there exists a \kl{run} from the \kl{state} $p$ on the current \kl{vertex} to $q$ on the same \kl{vertex}.

\begin{definition}\label{definition: loop-automata}
Let $\Sigma$ be a finite set of \kl{characters}.
A (non-deterministic) \intro*\kl{loop-automaton} $\automaton$ over $\Sigma$ is a \kl{graph} over $\LoopLabel{\Sigma, \univ{\automaton}}$ with $2$ ports, "ie",
a tuple $\tuple{\univ{\automaton}, \set{\alpha^{\automaton}}_{\alpha \in \LoopLabel{\Sigma, \univ{\automaton}}}, \src^{\automaton}, \tgt^{\automaton}}$ where
\begin{itemize}
    \item $\univ{\automaton}$ is a finite set of \kl{states},
    \item $\alpha^{\automaton} \subseteq \univ{\automaton}^2$ is a binary relation for expressing a transition function for each $\alpha \in \LoopLabel{\Sigma, \univ{\automaton}}$,
    \item $\intro*\src^{\automaton} \in \univ{\automaton}$ is the \kl{source} \kl{state},
    \item $\intro*\tgt^{\automaton} \in \univ{\automaton}$ is the \kl{target} \kl{state}.
\end{itemize}
\end{definition}

Given a \kl{loop-automaton} $\automaton$ and a \kl{structure} $\struc$ over a finite set $\Sigma$,
an $\struc$-\intro*\kl{run} (or, \kl{run} if $\struc$ is clear from the context) $\trace$ of $\automaton$ is
a non-empty sequence $x_0 \alpha_1 x_1 \dots x_{n-1} \alpha_{n} x_{n}$,
where
$x_i \in \univ{\struc}$ and $\alpha_i \in \LoopLabel{\Sigma, \univ{\automaton}}$ for each $i$.
We write $\src^{\trace} = x_0$ and $\tgt^{\trace} = x_n$.
For two $\struc$-\kl{runs} $\trace[1] = x_0 \alpha_1 \dots \alpha_{n} x_n$ and $\trace[2] = y_0 \beta_1 \dots \beta_{m} y_m$,
the \intro*\kl{coalesced product} $\trace[1] \series \trace[2]$ is partially defined as follows:
$\trace[1] \series \trace[2] \defeq x_0 \alpha_1 \dots \alpha_{n} x_n \beta_1 \dots \beta_{m} y_m$
if $x_n = y_0$ and undefined otherwise.
\AP
The relation $p \intro*\transrel{\struc, \automaton}{\trace} q$ (or, $p \transrel{\struc}{\trace} q$ if $\automaton$ is clear)
is defined as the minimal ternary relation of $p, q \in \univ{\automaton}$ and a \kl{run} $\trace$, closed under the following rules:
\begin{gather*}
    \begin{prooftree}
        \hypo{p \in \univ{\automaton} \mbox{ and } x \in \univ{\struc}}
        \infer1[R]{p \transrel{\struc}{x} p}
    \end{prooftree}
    \qquad
    \begin{prooftree}
        \hypo{\tuple{p, q} \in \alpha^{\automaton} \mbox{ and } \tuple{x, y} \in \jump{\alpha}^{\struc}}
        \infer1[$\alpha$]{p \transrel{\struc}{x \alpha y} q}
    \end{prooftree} \text{ for $\alpha \in \Sigma \dcup \set{\id}$}
    \\
    \begin{prooftree}
        \hypo{p \transrel{\struc}{\trace[1]} q}
        \hypo{q \transrel{\struc}{\trace[2]} r}
        \infer2[T]{p \transrel{\struc}{\trace[1] \series \trace[2]} r}
    \end{prooftree}
    \qquad
    \begin{prooftree}
        \hypo{\tuple{p, q} \in \ell_{\tuple{p', q'}}^{\automaton} \mbox{, } p' \transrel{\struc}{\trace[1]} q' \mbox{, and } \src^{\trace[1]} = \tgt^{\trace[1]} = x \in \univ{\struc}}
        \infer1[$\ell_{\tuple{p', q'}}$]{p \transrel{\struc}{x \ell_{\tuple{p', q'}} x} q}
    \end{prooftree}
\end{gather*}
The \kl{semantics} $\jump{\automaton}^{\struc} \subseteq \univ{\struc}^2$ is given as follows:
\[\jump{\automaton}^{\struc} ~\defeq~ \set{\tuple{x, y} \in \univ{\struc}^2 \mid \src^{\automaton} \transrel{\struc}{\trace[1]} \tgt^{\automaton} \mbox{ for some \kl{run} $\trace[1]$ "st" $\src^{\trace[1]} = x$ and $\tgt^{\trace[1]} = y$}}.\]

\begin{example}\label{example: loop-automata}
Let $\term \defeq a a (b a)^{\lop} b b$.
Then the following $\loopRKAautomaton{\term}$ is the corresponding \kl{loop-automaton} (i.e., $\jump{\loopRKAautomaton{\term}}^{\struc} = \jump{\term}^{\struc}$ holds for every \kl{structure} $\struc$):

\[\loopRKAautomaton{\term} \defeq \left(\begin{tikzpicture}[baseline = -2.5ex]
        \graph[grow right = 1.9cm, branch down = 4ex, nodes={}]{
        {/,0/{$p_1$}[mynode, draw, circle]}-!-{0'/{$q_1$}[mynode, draw, circle, xshift = 0.9cm], 1/{$p_2$}[mynode, draw, circle]}-!-{1'/{$q_2$}[mynode, draw, circle, , xshift = 0.9cm], 2/{$p_3$}[mynode, draw, circle]}-!-{2'/{$q_3$}[mynode, draw, circle, xshift = 0.9cm],3/{$p_4$}[mynode, draw, circle]}-!-{/,4/{$p_5$}[mynode, draw, circle]}-!-{/,5/{$p_6$}[mynode, draw, circle]}
        };
        \node[left = .5em of 0](l){};
        \node[right = .5em of 5](r){};
        \graph[use existing nodes, edges={color=black, pos = .5, earrow}, edge quotes={fill=white, inner sep=1pt,font= \scriptsize}]{
        0' ->["$b$"] 1' ->["$a$"] 2';
        0 ->["$a$"] 1 ->["$a$"] 2 ->["$\ell_{\tuple{q_1, q_3}}$"] 3 ->["$b$"] 4 ->["$b$"] 5;
        l -> 0; 5 -> r;
        };
    \end{tikzpicture}\right).\]
    For instance, let $\struc \defeq \glue (\textcolor{red!70}{\struc[2]_1} \textcolor{blue!50}{\struc[2]_2})$ where
    $\textcolor{red!70}{\struc[2]_1} = \left(\begin{tikzpicture}[baseline = -.5ex]
        \graph[grow right = 1.cm, branch down = 2.5ex]{
        {s1/{$1$}[vert, text = red!70]} -!- {t1/{$2$}[vert]}
        };
        \graph[use existing nodes, edges={color=black, pos = .5, earrow}, edge quotes={fill=white, inner sep=1pt,font= \scriptsize}]{
            s1 ->["$a$", bend right] t1;
            t1 ->["$b$", bend right] s1;
        };
    \end{tikzpicture} \right)$
     and 
    $\textcolor{blue!50}{\struc[2]_2} = \left(\begin{tikzpicture}[baseline = -.5ex]
        \graph[grow right = 1.cm, branch down = 2.5ex]{
        {s1/{$2$}[vert]} -!- {t1/{$3$}[vert, text = blue!50]}
        };
        \graph[use existing nodes, edges={color=black, pos = .5, earrow}, edge quotes={fill=white, inner sep=1pt,font= \scriptsize}]{
            s1 ->["$a$", bend right] t1;
            t1 ->["$b$", bend right] s1;
        };
    \end{tikzpicture} \right)$.
    We just write
    $\quo{\tuple{1, 1}}_{\sim}$,
    $\quo{\tuple{1, 2}}_{\sim}$ ($= \quo{\tuple{2, 2}}_{\sim}$), and
    $\quo{\tuple{2, 3}}_{\sim}$,
    as $\tilde{1}$, $\tilde{2}$, and $\tilde{3}$, respectively, for short.
    We then have $\tuple{\tilde{1}, \tilde{1}} \in \jump{\term}^{\struc}$.
    We also have $\tuple{\tilde{1}, \tilde{1}} \in \jump{\loopRKAautomaton{\term}}^{\struc}$,
    "eg", by the following derivation tree:
    \begin{center}
        \scalebox{1}{\begin{prooftree}[separation = .6em]
            \hypo{p_1 \transrel{\struc}{\tilde{\textcolor{red!70}{1}} a \tilde{2}} p_2}
            \hypo{p_2 \transrel{\struc}{\tilde{2} a \tilde{\textcolor{blue!50}{3}}} p_3}
            \infer2[T]{p_1 \transrel{\struc}{\tilde{\textcolor{red!70}{1}} a \tilde{2} a \tilde{\textcolor{blue!50}{3}}} p_3}
            \hypo{q_1 \transrel{\struc}{\tilde{\textcolor{blue!50}{3}} b \tilde{2}} q_2}
            \hypo{q_2 \transrel{\struc}{\tilde{2} a \tilde{\textcolor{blue!50}{3}}} q_3}
            \infer2[T]{q_1 \transrel{\struc}{\tilde{\textcolor{blue!50}{3}} b \tilde{2} a \tilde{\textcolor{blue!50}{3}}} q_3}
            \infer1[$\ell_{\tuple{q_1, q_3}}$]{p_3 \transrel{\struc}{\tilde{\textcolor{blue!50}{3}} \ell_{\tuple{q_1, q_3}} \tilde{\textcolor{blue!50}{3}}} p_4}
            \hypo{p_4 \transrel{\struc}{\tilde{\textcolor{blue!50}{3}} b \tilde{2}} p_5}
            \hypo{p_5 \transrel{\struc}{\tilde{2} b \tilde{\textcolor{red!70}{1}}} p_6}
            \infer2[T]{p_4 \transrel{\struc}{\tilde{\textcolor{blue!50}{3}} b \tilde{2} b \tilde{\textcolor{red!70}{1}}} p_6}
            \infer[double]3[T]{p_1 \transrel{\struc}{\tilde{\textcolor{red!70}{1}} a \tilde{2} a \tilde{3} \ell_{\tuple{q_1, q_3}} \tilde{3} b \tilde{2} b \tilde{\textcolor{red}{1}}} p_6}
        \end{prooftree}}
    \end{center}
\end{example}
Generalizing \Cref{example: loop-automata},
we can transform every \kl{loop-RKA term} into a \kl{loop-automaton}, by extending the McNaughton--Yamada--Thompson construction \cite{mcnaughtonRegularExpressionsState1960,thompsonProgrammingTechniquesRegular1968}
for the \kl{graph loop} $\term^{\lop}$.
\knowledgeconfigure {quotation=false}
\begin{defi}\label{defi: term to automata}
    For a \kl{loop-RKA term} $\term$,
    we define the \kl{loop-automaton} $\intro*\loopRKAautomaton{\term}$, as follows:\footnote{%
For simplicity, this graphical definition is given modulo renaming of \kl{states} by hiding names (note that $\ell_{\tuple{p, q}}$ also uses names of \kl{states}).}
    \begin{gather*}
        \loopRKAautomaton{a} \defeq \begin{tikzpicture}[baseline = -.5ex]
            \graph[grow right = 1.cm, branch down = 6ex, nodes={mynode}]{
            {0/{}[draw, circle]}-!-{1/{}[draw, circle]}
            };
            \node[left = .5em of 0](l){};
            \node[right = .5em of 1](r){};
            \graph[use existing nodes, edges={color=black, pos = .5, earrow}, edge quotes={fill=white, inner sep=1pt,font= \scriptsize}]{
            0 ->["$a$"] 1;
            l -> 0; 1 -> r;
            };
        \end{tikzpicture}, \quad
        \loopRKAautomaton{\emp} \defeq \begin{tikzpicture}[baseline = -.5ex]
            \graph[grow right = 1.cm, branch down = 6ex, nodes={mynode}]{
            {0/{}[draw, circle]}-!-{1/{}[draw, circle]}
            };
            \node[left = .5em of 0](l){};
            \node[right = .5em of 1](r){};
            \graph[use existing nodes, edges={color=black, pos = .5, earrow}, edge quotes={fill=white, inner sep=1pt,font= \scriptsize}]{
            l -> 0; 1 -> r;
            };
        \end{tikzpicture}, \quad
        \loopRKAautomaton{\id} \defeq \begin{tikzpicture}[baseline = -.5ex]
            \graph[grow right = 1.cm, branch down = 6ex, nodes={mynode}]{
            {0/{}[draw, circle]}-!-{1/{}[draw, circle]}
            };
            \node[left = .5em of 0](l){};
            \node[right = .5em of 1](r){};
            \graph[use existing nodes, edges={color=black, pos = .5, earrow}, edge quotes={fill=white, inner sep=1pt,font= \scriptsize}]{
            0 ->["$\id$"] 1;
            l -> 0; 1 -> r;
            };
        \end{tikzpicture},\\ 
        \loopRKAautomaton{\term[1] \union \term[2]} \defeq \begin{tikzpicture}[baseline = -2.5ex]
            \graph[grow right = 1.cm, branch down = 2ex, nodes={}]{
            {/,0/{}[mynode, draw, circle]}-!-{11/{}[mynode, draw, circle],/,12/{}[mynode, draw, circle]}-!-{21/{}[mynode, draw, circle],/,22/{}[mynode, draw, circle]}-!-{/,3/{}[mynode, draw, circle]}
            };
            \node[left = .5em of 0](l){};
            \node[right = .5em of 3](r){};
            \graph[use existing nodes, edges={color=black, pos = .5, earrow}, edge quotes={fill=white, inner sep=1pt,font= \scriptsize}]{
            0 ->["$\id$"] {11, 12};
            {21,22} ->["$\id$"] {3};
            {11} ->["$\loopRKAautomaton{\term[1]}$"{draw,rectangle}] {21};
            {12} ->["$\loopRKAautomaton{\term[2]}$"{draw,rectangle}] {22};
            l -> 0; 3 -> r;
            };
        \end{tikzpicture}, \quad
                \loopRKAautomaton{\term[1] \compo \term[2]} \defeq \begin{tikzpicture}[baseline = -.5ex]
            \graph[grow right = 1.cm, branch down = 6ex, nodes={mynode}]{
            {0/{}[draw, circle]}-!-{1/{}[draw, circle]}-!-{2/{}[draw, circle]}
            };
            \node[left = .5em of 0](l){};
            \node[right = .5em of 2](r){};
            \graph[use existing nodes, edges={color=black, pos = .5, earrow}, edge quotes={fill=white, inner sep=1pt,font= \scriptsize}]{
            0 ->["$\loopRKAautomaton{\term[1]}$"{draw,rectangle}] 1;
            1 ->["$\loopRKAautomaton{\term[2]}$"{draw,rectangle}] 2;
            l -> 0; 2 -> r;
            };
        \end{tikzpicture}, \\
        \loopRKAautomaton{\term[1]^*} \defeq \begin{tikzpicture}[baseline = -.5ex]
            \graph[grow right = 1.cm, branch down = 2ex, nodes={}]{
            {0/{}[mynode, draw, circle]}-!-{1/{}[mynode, draw, circle]}-!-{2/{}[mynode, draw, circle]}-!-{3/{}[mynode, draw, circle]}
            };
            \node[left = .5em of 0](l){};
            \node[right = .5em of 3](r){};
            \graph[use existing nodes, edges={color=black, pos = .5, earrow}, edge quotes={fill=white, inner sep=1pt,font= \scriptsize}]{
            0 ->["$\id$"] 1 ->["$\loopRKAautomaton{\term}$"{draw,rectangle}] 2 ->["$\id$"] 3;
            2 ->["$\id$", bend right = 60] 1;
            l -> 0; 3 -> r;
            0 ->["$\id$", bend right = 20, out = -45, in = -135, looseness = .5] 3;
            };
        \end{tikzpicture}, \quad
        \loopRKAautomaton{\term^{\lop}} \defeq \begin{tikzpicture}[baseline = -2.ex]
            \graph[grow right = 1.5cm, branch down = 3ex, nodes={}]{
            {u0/{$p'$}[mynode, draw, circle], 0/{}[mynode, draw, circle]}-!-{u1/{$q'$}[mynode, draw, circle], 1/{}[mynode, draw, circle]}
            };
            \node[left = .5em of 0](l){};
            \node[right = .5em of 1](r){};
            \graph[use existing nodes, edges={color=black, pos = .5, earrow}, edge quotes={fill=white, inner sep=1pt,font= \scriptsize}]{
            u0 ->["$\loopRKAautomaton{\term}$"{draw,rectangle}] u1;
            0 ->["$\ell_{\tuple{p', q'}}$"] 1;
            l -> 0; 1 -> r;
            };
        \end{tikzpicture} \mbox{where $p', q'$ are \kl{fresh}.}
    \end{gather*}
\end{defi}
\knowledgeconfigure {quotation}
The construction is computable in time $\mathcal{O}(\|\term\|)$ and produces
$\mathcal{O}(\|\term\|)$ states and transitions.

\begin{proposition}\label{prop: automata}
    For every \kl(loop-RKA){term} $\term$ and every \kl{structure} $\struc$,
    we have $\jump{\term}^{\struc} = \jump{\loopRKAautomaton{\term}}^{\struc}$.
\end{proposition}
\begin{proof}
    By easy induction on $\term$.
    As the other cases are the same as the McNaughton--Yamada--Thompson construction,
    we only write the case of $\term[1] = \term[2]^{\lop}$.
    First, we have both $\jump{\term[2]^{\lop}}^{\struc} \subseteq \diagonal_{\univ{\struc}}$ and $\jump{\loopRKAautomaton{\term[2]^{\lop}}}^{\struc} \subseteq \diagonal_{\univ{\struc}}$ by the form of $\term[2]^{\lop}$ and $\loopRKAautomaton{\term[2]^{\lop}}$, respectively.
    We also have:
    $\tuple{x, x} \in \jump{\term[2]^{\lop}}^{\struc}$ "iff"
    $\tuple{x, x} \in \jump{\term[2]}^{\struc}$ "iff"
    $\tuple{x, x} \in \jump{\loopRKAautomaton{\term[2]}}^{\struc}$ (IH) "iff"
    there is a \kl{run} $\trace$ of $\loopRKAautomaton{\term[2]}$ on $\struc$ such that $\src^{\trace} = \tgt^{\trace} = x$
    and $\src^{\loopRKAautomaton{\term[2]}} \transrel{\struc}{\trace} \tgt^{\loopRKAautomaton{\term[2]}}$
    "iff"
    $\src^{\loopRKAautomaton{\term[2]^{\lop}}}
    \transrel{\struc}{x \ell_{\tuple{\src^{\loopRKAautomaton{\term[2]}}, \tgt^{\loopRKAautomaton{\term[2]}}}} x}
    \tgt^{\loopRKAautomaton{\term[2]^{\lop}}}$ (by construction of $\loopRKAautomaton{\term[2]^{\lop}}$) "iff"
    $\tuple{x, x} \in \jump{\loopRKAautomaton{\term[2]^{\lop}}}^{\struc}$.
    Hence, this case has been proved.
\end{proof}
From \Cref{prop: automata},
in the sequel,
we use \kl{loop-automata} instead of \kl(loop-RKA){terms} in \kl{loop-RKA}.
\section{Decomposition Theorem}\label{section: decomposing derivatives}
Let $\mul{\struc} = \struc_1 \dots \struc_n \in \REL^{+}$.
We consider the glued \kl{structure} $\glue \mul{\struc}$ (recall \Cref{section: preliminaries}).
For each $i \in \rangeone{n}$, we define the $3$-partition $\set{\intro*\univc{\mul{\struc}}{i}{}, \intro*\univl{\mul{\struc}}{i}{}, \intro*\univr{\mul{\struc}}{i}{}}$ of $\univ{\glue \mul{\struc}}$ (whose parts may be empty):
\begin{align*}
\univc{\mul{\struc}}{i}{} &\defeq \set{[\tuple{i, \lab}]_{\sim_{\glue \mul{\struc}}} \mid \lab \in \univ{\mul{\struc}\getith{i}}},\\ 
\univl{\mul{\struc}}{i}{} &\defeq (\bigcup_{j = 1}^{i} \univc{\mul{\struc}}{j}{}) \setminus \univc{\mul{\struc}}{i}{},& \hspace{-.5em} 
\univr{\mul{\struc}}{i}{} &\defeq (\bigcup_{j = i}^{n} \univc{\mul{\struc}}{j}{}) \setminus \univc{\mul{\struc}}{i}{}.
\end{align*}
In this section, we introduce the following relation $\decomposedtransrel{\mul{\struc}, \automaton}{\trace}$.
\begin{defi}\label{definition: decomposition}
    Let $\mul{\struc} = \struc_1 \dots \struc_n \in \REL^{+}$.
    The relation $p \intro*\decomposedtransrel{\mul{\struc}, \automaton}{\trace} q$ (or, $p \decomposedtransrel{\mul{\struc}}{\trace} q$ if $\automaton$ is clear)
    is defined as the minimal ternary relation of $p, q \in \univ{\automaton}$ and an $\glue \mul{\struc}$-\kl{run} $\trace$
    such that $\tuple{\src^{\trace}, \tgt^{\trace}} \in \bigcup_{i = 1}^n \univc{\mul{\struc}}{i}{2}$,
    closed under the following rules:
    \begin{gather*}
        \scalebox{.93}{\begin{prooftree}
            \hypo{p \in \univ{\automaton} \mbox{ and } x \in \univ{\mul{\struc}\getith{i}}}
            \infer1[R]{p \decomposedtransrel{\mul{\struc}}{\quo{\tuple{i,x}}_{\sim}} p}
        \end{prooftree}}
        \quad
        \scalebox{.93}{\begin{prooftree}
            \hypo{\tuple{p, q} \in \alpha^{\automaton} \mbox{ and } \tuple{x, y} \in \jump{\alpha}^{\mul{\struc}\getith{i}}}
            \infer1[$\alpha$]{p \decomposedtransrel{\mul{\struc}}{\quo{\tuple{i, x}}_{\sim} \alpha \quo{\tuple{i, y}}_{\sim}} q}
        \end{prooftree}}
        \text{ for $\alpha \in \Sigma \dcup \set{\id}$}
        \\
        \scalebox{.93}{\begin{prooftree}
            \hypo{p \decomposedtransrel{\mul{\struc}}{\trace[1]} q}
            \hypo{q \decomposedtransrel{\mul{\struc}}{\trace[2]} r}
            \infer2[T]{p \decomposedtransrel{\mul{\struc}}{\trace[1] \series \trace[2]} r}
        \end{prooftree}}
        \quad
        \scalebox{.93}{\begin{prooftree}
            \hypo{\tuple{p, q} \in \ell_{\tuple{p', q'}}^{\automaton}, p' \decomposedtransrel{\mul{\struc}}{\trace[1]} q', \mbox{ and } \src^{\trace[1]} = \tgt^{\trace[1]} = \quo{\tuple{i, x}}_{\sim} \in \univ{\glue \mul{\struc}}}
            \infer1[$\ell_{\tuple{p', q'}}$]{p \decomposedtransrel{\mul{\struc}}{\quo{\tuple{i, x}}_{\sim} \ell_{\tuple{p', q'}} \quo{\tuple{i, x}}_{\sim}} q}
        \end{prooftree}}
    \end{gather*}
\end{defi}
The relation $\decomposedtransrel{\mul{\struc}}{\trace}$
is a ``decomposed'' version of $\transrel{\glue \mul{\struc}}{\trace}$
in that \kl{runs} $\trace$ in $\decomposedtransrel{\mul{\struc}}{\trace}$ are
constructed by composing \kl{runs} such that $\tuple{\src^{\trace}, \tgt^{\trace}} \in \bigcup_{i = 1}^n \univc{\mul{\struc}}{i}{2}$.
The following theorem shows that $\decomposedtransrel{\mul{\struc}}{\trace}$ coincides with $\transrel{\glue \mul{\struc}}{\trace}$
(under the restriction of \kl{runs} $\trace$ that $\tuple{\src^{\trace}, \tgt^{\trace}} \in \bigcup_{i = 1}^n \univc{\mul{\struc}}{i}{2}$).
\begin{restatable}[Decomposition theorem]{theorem}{decompositiontheorem}\label{theorem: decomposition}
    Let $\automaton$ be a \kl{loop-automaton} and let $\mul{\struc} = \struc_1 \dots \struc_n \in \REL^{+}$.
    For all $p, q \in \univ{\automaton}$ and \kl{runs} $\trace$ on $\glue\mul{\struc}$ such that $\tuple{\src^{\trace}, \tgt^{\trace}} \in \bigcup_{i = 1}^n \univc{\mul{\struc}}{i}{2}$,
    we have:
    \[p \transrel{\glue \mul{\struc}}{\trace} q \quad\iff\quad p \decomposedtransrel{\mul{\struc}}{\trace} q.\]
\end{restatable}
The direction $(\Longleftarrow)$ is easy by using the derivation tree of the same form.
The converse direction $(\Longrightarrow)$ is relatively non-trivial.
We first give an example.
\begin{exa}\label{example: decomposition}
    Recall the derivation tree given in \Cref{example: loop-automata}.
    In the tree, we cannot replace ($\transrel{\glue \mul{\struc[2]}}{}$) with ($\decomposedtransrel{\mul{\struc[2]}}{}$)
    straightforwardly,
    because some \kl{runs} contain the pair $\tilde{\textcolor{red!70}{1}}$ and $\tilde{\textcolor{blue!50}{3}}$ 
    as the \kl{source} and \kl{target}
    (note that there is no $i$ such that $\set{\tilde{\textcolor{red!70}{1}}, \tilde{\textcolor{blue!50}{3}}} \subseteq \univc{\mul{\struc[2]}}{i}{}$).
    Nevertheless, we can show $\tuple{\tilde{1}, \tilde{1}} \in \jump{\automaton}^{\glue \mul{\struc[2]}}$ by modifying the derivation tree as follows:
    \begin{center}
        \scalebox{1}{\begin{prooftree}[separation = .5em]
            \hypo{p_1 \decomposedtransrel{\mul{\struc[2]}}{\tilde{\textcolor{red!70}{1}} a \tilde{2}} p_2}
            \hypo{p_2 \decomposedtransrel{\mul{\struc[2]}}{\tilde{2} a \tilde{\textcolor{blue!50}{3}}} p_3}
            \hypo{q_1 \decomposedtransrel{\mul{\struc[2]}}{\tilde{\textcolor{blue!50}{3}} b \tilde{2}} q_2}
            \hypo{q_2 \decomposedtransrel{\mul{\struc[2]}}{\tilde{2} a \tilde{\textcolor{blue!50}{3}}} q_3}
            \infer2[T]{q_1 \decomposedtransrel{\mul{\struc[2]}}{\tilde{\textcolor{blue!50}{3}} b \tilde{2} a \tilde{\textcolor{blue!50}{3}}} q_3}
            \infer1[$\ell_{\tuple{q_1, q_3}}$]{p_3 \decomposedtransrel{\mul{\struc[2]}}{\tilde{\textcolor{blue!50}{3}} \ell_{\tuple{q_1, q_3}} \tilde{\textcolor{blue!50}{3}}} p_4}
            \hypo{p_4 \decomposedtransrel{\mul{\struc[2]}}{\tilde{\textcolor{blue!50}{3}} b \tilde{2}} p_5}
            \infer[double]3[T]{p_2 \decomposedtransrel{\mul{\struc[2]}}{\tilde{2} a \tilde{3} \ell_{\tuple{q_1, q_3}} \tilde{3} b \tilde{2}} p_5}
            \hypo{p_5 \decomposedtransrel{\mul{\struc[2]}}{\tilde{2} b \tilde{\textcolor{red!70}{1}}} p_6}
            \infer[double]3[T]{p_1 \decomposedtransrel{\mul{\struc[2]}}{\tilde{\textcolor{red!70}{1}} a \tilde{2} a \tilde{3} \ell_{\tuple{q_1, q_3}} \tilde{3} b \tilde{2} b \tilde{\textcolor{red!70}{1}}} p_6}
        \end{prooftree}}
        \end{center}
\end{exa}

\subsection*{Proof of \Cref{theorem: decomposition}}
We first observe the following lemma.
\begin{lem}\label{lem: decomposing traces}
    Let $\automaton$ be a \kl{loop-automaton} and let $\struc \in \REL$.
    Let $p, q \in \univ{\automaton}$ and let $\trace = \tilde{x}_0 \alpha_1 \tilde{x}_1 \dots \tilde{x}_{m-1} \alpha_m \tilde{x}_{m}$.
    If $p \transrel{\struc}{\trace} q$,
    then there exist $r_0, \dots, r_{m}$ such that $r_0 = p$, $r_{m} = q$, and for each $j \in \rangeone{m}$,
    we have $r_{j-1} \transrel{\struc}{\tilde{x}_{j-1} \alpha_{j} \tilde{x}_{j}} r_{j}$.
\end{lem}
\begin{proof}
    By easy induction on the derivation tree of $p \transrel{\struc}{\trace} q$.
    More precisely, the same induction chooses one-step derivations $D_1,\dots,D_m$ with $\sum_j |D_j|+(m-1)\le |D|$, where $D$ is the original derivation; hence, when $m\ge2$, every non-empty proper consecutive subsequence composes to a derivation smaller than $D$.
\end{proof}
For instance, for $p_1 \transrel{\glue \mul{\struc[2]}}{\tilde{\textcolor{red!70}{1}} a \tilde{2} a \tilde{3} \ell_{\tuple{q_1, q_3}} \tilde{3} b \tilde{2} b \tilde{\textcolor{red!70}{1}}} p_6$ in \Cref{example: loop-automata},
we can take the following sequence:
$p_1 \transrel{\glue \mul{\struc[2]}}{\tilde{\textcolor{red!70}{1}} a \tilde{2}} p_2
\transrel{\glue \mul{\struc[2]}}{\tilde{2} a \tilde{\textcolor{blue!50}{3}}} p_3
\transrel{\glue \mul{\struc[2]}}{\tilde{\textcolor{blue!50}{3}} \ell_{\tuple{q_1, q_3}} \tilde{\textcolor{blue!50}{3}}} p_4
\transrel{\glue \mul{\struc[2]}}{\tilde{\textcolor{blue!50}{3}} b \tilde{2}} p_5
\transrel{\glue \mul{\struc[2]}}{\tilde{2} b \tilde{\textcolor{red!70}{1}}} p_6$.

\begin{proof}[Proof of \Cref{theorem: decomposition}]
    ($\Longleftarrow$):
    By using the derivation tree of the same form.
    ($\Longrightarrow$):
    By induction on the size of the derivation tree of $p \transrel{\glue \mul{\struc}}{\trace} q$.
    By the assumption on $\trace$, fix $i \in \rangeone{n}$ such that $\tuple{\src^{\trace}, \tgt^{\trace}} \in \univc{\mul{\struc}}{i}{2}$.
    In the induction step, we use the following one-step observation.
    If $s \transrel{\glue \mul{\struc}}{\tilde{u} \alpha \tilde{v}} t$ and $\tuple{\tilde{u}, \tilde{v}} \in \univc{\mul{\struc}}{h}{2}$ for some $h$,
    then $s \decomposedtransrel{\mul{\struc}}{\tilde{u} \alpha \tilde{v}} t$.
    Indeed, when $\alpha \in \Sigma \dcup \set{\id}$, this follows from the corresponding one-step rule in \Cref{definition: decomposition}, using a bag witnessing the edge or the diagonal.
    When $\alpha = \ell_{\tuple{s', t'}}$, take a smallest derivation of $s \transrel{\glue \mul{\struc}}{\tilde{u} \alpha \tilde{v}} t$.
    Its last rule must be the loop-label rule, and hence it gives a smaller derivation
    $s' \transrel{\glue \mul{\struc}}{\trace[1]} t'$ with $\src^{\trace[1]} = \tgt^{\trace[1]} = \tilde{u} = \tilde{v}$;
    by IH, $s' \decomposedtransrel{\mul{\struc}}{\trace[1]} t'$, and then the loop-label rule in \Cref{definition: decomposition} gives the claim.
    We distinguish the following cases:
    \begin{itemize}
        \item (Base Case) Case $\trace = \tilde{x}$ and Case $\trace = \tilde{x} \alpha \tilde{y}$ where $\alpha \in \LoopLabel{\Sigma, \univ{\automaton}}$:
        Easy by applying the same rule and the one-step observation above.
        \item (Inductive Case) Case $\trace = \tilde{x}_0 \alpha_1 \tilde{x}_1 \dots \tilde{x}_{m-1} \alpha_{m} \tilde{x}_{m}$
        where $m \ge 2$ and $\alpha_{i} \in \LoopLabel{\Sigma, \univ{\automaton}}$ for each $i \in \rangeone{m}$.
        Choose $x_0, x_m \in \univ{\mul{\struc}\getith{i}}$ such that
        $\tilde{x}_0 = \quo{\tuple{i, x_0}}_{\sim}$ and $\tilde{x}_m = \quo{\tuple{i, x_m}}_{\sim}$.
        By \Cref{lem: decomposing traces},
        there exist $r_0, \dots, r_{m}$ such that
        $r_0 = p$,
        $r_{m} = q$, and
        $r_{j-1} \transrel{\glue \mul{\struc}}{\tilde{x}_{j-1} \alpha_{j} \tilde{x}_{j}} r_{j}$ for each $j$.
        We distinguish the following cases:
        \begin{itemize}
            \item Case $\tilde{x}_1 \in \univc{\mul{\struc}}{i}{}$:
            We then have:
            \begin{center}
                \scalebox{1}{\begin{prooftree}
                    \hypo{\mathstrut}
                    \infer1[{\footnotesize one-step}]{r_0 \decomposedtransrel{\mul{\struc}}{\tilde{x}_0 \alpha_1 \tilde{x}_1} r_1}
                    \hypo{\mathstrut}
                    \infer[dashed]1[IH at $i$]{r_1 \decomposedtransrel{\mul{\struc}}{\tilde{x}_1 \alpha_2 \dots \alpha_{m} \tilde{x}_{m}} r_m}
                    \infer2[T]{r_0 \decomposedtransrel{\mul{\struc}}{\tilde{x}_0 \alpha_1 \dots \alpha_{m} \tilde{x}_{m}} r_m}
                \end{prooftree}}
            \end{center}
            \item Case $\tilde{x}_{m-1} \in \univc{\mul{\struc}}{i}{}$:
            Similar to the case above.
            \item Case $\tuple{\tilde{x}_1, \tilde{x}_{m-1}} \in \univl{\mul{\struc}}{i}{2} \dcup \univr{\mul{\struc}}{i}{2}$:
            WLOG, we consider when $\tuple{\tilde{x}_1, \tilde{x}_{m-1}} \in \univr{\mul{\struc}}{i}{2}$.
            The first and last one-step transitions have witnesses $i_1, i_m > i$ such that
            $\tilde{x}_{0} = \quo{\tuple{i_1, x_{0}}}_{\sim}$,
            $\tilde{x}_{1} = \quo{\tuple{i_1, x_{1}}}_{\sim}$,
            $\tilde{x}_{m-1} = \quo{\tuple{i_m, x_{m-1}}}_{\sim}$, and
            $\tilde{x}_{m} = \quo{\tuple{i_m, x_m}}_{\sim}$.
            \begin{itemize}
                \item Case $i_{1} \le i_m$:
                By the definition of $\sim_{\glue \mul{\struc}}$ with $\tilde{x}_{m} = \quo{\tuple{i_m, x_m}}_{\sim} = \quo{\tuple{i, x_m}}_{\sim}$ and $i \le i_1 \le i_m$,
                we have $\tilde{x}_{m} = \quo{\tuple{i_{1}, x_m}}_{\sim}$.
                As $\tilde{x}_{1}, \tilde{x}_m \in \univc{\mul{\struc}}{i_1}{}$, we have:
                \begin{center}
                    \scalebox{1}{\begin{prooftree}
                        \hypo{\mathstrut}
                        \infer1[{\footnotesize one-step}]{r_0 \decomposedtransrel{\mul{\struc}}{\tilde{x}_0 \alpha_1 \tilde{x}_1} r_1}
                        \hypo{\mathstrut}
                        \infer[dashed]1[IH at $i_1$]{r_1 \decomposedtransrel{\mul{\struc}}{\tilde{x}_1 \alpha_{2} \dots \alpha_{m} \tilde{x}_{m}} r_{m}}
                        \infer2[T]{r_0 \decomposedtransrel{\mul{\struc}}{\tilde{x}_0 \alpha_1 \dots \alpha_{m} \tilde{x}_{m}} r_{m}}
                    \end{prooftree}}
                \end{center}
                \item Case $i_{1} \ge i_m$:
                Similarly, we have $\tilde{x}_{0}, \tilde{x}_{m-1} \in \univc{\mul{\struc}}{i_m}{}$,
                and thus we have:
                \begin{center}
                    \scalebox{1}{\begin{prooftree}
                        \hypo{\mathstrut}
                        \infer[dashed]1[IH at $i_m$]{r_0 \decomposedtransrel{\mul{\struc}}{\tilde{x}_0 \alpha_{1} \dots \alpha_{m-1} \tilde{x}_{m-1}} r_{m-1}}
                        \hypo{\mathstrut}
                        \infer1[{\footnotesize one-step}]{r_{m-1} \decomposedtransrel{\mul{\struc}}{\tilde{x}_{m-1} \alpha_m \tilde{x}_{m}} r_m}
                        \infer2[T]{r_0 \decomposedtransrel{\mul{\struc}}{\tilde{x}_0 \alpha_1 \dots \alpha_{m} \tilde{x}_{m}} r_{m}}
                    \end{prooftree}}
                \end{center}
              
            \end{itemize}
            \item Otherwise, $\tuple{\tilde{x}_1, \tilde{x}_{m-1}} \in (\univl{\mul{\struc}}{i}{} \times \univr{\mul{\struc}}{i}{}) \dcup (\univr{\mul{\struc}}{i}{} \times \univl{\mul{\struc}}{i}{})$:
            We observe that $\set{\univl{\mul{\struc}}{i}{}, \univc{\mul{\struc}}{i}{}, \univr{\mul{\struc}}{i}{}}$ is a $3$-partition of the set $\univ{\glue \mul{\struc}}$
            and $\tuple{\tilde{x}_{j-1}, \tilde{x}_{j}} \not\in (\univl{\mul{\struc}}{i}{} \times \univr{\mul{\struc}}{i}{}) \dcup (\univr{\mul{\struc}}{i}{} \times \univl{\mul{\struc}}{i}{})$ for each $j$.
            Then by an analog of the ``discrete intermediate value theorem'',
            there exists some $1 < k < m-1$ such that $\tilde{x}_k \in \univc{\mul{\struc}}{i}{}$.
            We then have:
            \begin{center}
                \begin{prooftree}
                    \hypo{\mathstrut}
                    \infer[dashed]1[IH at $i$]{r_0 \decomposedtransrel{\mul{\struc}}{\tilde{x}_0 \alpha_1 \dots \alpha_{k} \tilde{x}_k} r_k}
                    \hypo{\mathstrut}
                    \infer[dashed]1[IH at $i$]{r_k \decomposedtransrel{\mul{\struc}}{\tilde{x}_k \alpha_{k+1} \dots \alpha_{m} \tilde{x}_{m}} r_m}
                    \infer2[T]{r_0 \decomposedtransrel{\mul{\struc}}{\tilde{x}_0 \alpha_1 \tilde{x}_1 \dots \tilde{x}_{m-1} \alpha_{m} \tilde{x}_{m}} r_m}
                \end{prooftree}
            \end{center}
        \end{itemize}
    \end{itemize}
    Hence, this completes the proof. 
\end{proof}

\section{Automata Construction} \label{section: automata construction}
In this section, 
we show that \Cref{theorem: decomposition}
induces a reduction
from the (\kl{relational semantics}) inclusion of \kl{loop-automata}
to the (\kl{string language}) inclusion of \kl{2-way alternating finite string automata} (\kl{2AFAs}).
Our reduction is based on \cite[\S 5.3]{nakamuraDerivativesGraphsPositive2025},
but the main difference is that we apply a binary encoding (\Cref{section: binary encoding}) for obtaining the \kl{PSPACE} upper bound.

\subsection{2AFA}
\AP
We use $\intro*\lanchor$ and $\intro*\ranchor$ as the special \kl{characters} denoting the leftmost and rightmost anchors.
A \intro*\kl{2AFA} $\automaton$ over a finite set $A$ is a tuple $\automaton = \tuple{\univ{\automaton}, \intro*\autotrans^{\automaton}, \intro*\autosrc^{\automaton}}$, where 
\begin{itemize}
    \item $\univ{\automaton}$ is a finite set of "states",
    \item \AP $\autotrans^{\automaton} \colon \univ{\automaton} \times (A \dcup \set{\lanchor, \ranchor}) \to \PBFML(\univ{\automaton} \times \set{-1, 0, 1})$ is a \emph{transition function},
    where $\intro*\PBFML(X)$ denotes the set of ""positive boolean formulas"" over a set $X$ of ""propositional variables"" given by
    \begin{align*}
        \fml[1], \fml[2] \in \PBFML(X) \Coloneqq p \mid \const{false} \mid \const{true} \mid \fml[1] \lor \fml[2] \mid \fml[1] \land \fml[2] \tag*{where $p \in X$,}
    \end{align*}
\item $\autosrc^{\automaton} \in \univ{\automaton}$ is the initial "state".
\end{itemize}
For a \kl{2AFA} $\automaton$ and a \kl{string} $w = a_{1} \dots a_{n}$ over $A \dcup \set{\lanchor, \ranchor}$,
the set $S^{\automaton}_{\word} \subseteq \univ{\automaton} \times \rangeone{n}$ is defined as
the smallest set closed under the following rule:
    For each $\tuple{q, i} \in \univ{\automaton} \times \rangeone{n}$
    and "propositional variables" $\tuple{q_1, i_1}, \dots, \tuple{q_{m}, i_{m}} \in \univ{\automaton} \times \set{-1, 0, 1}$,
    when the "positive boolean formula" $\autotrans^{\automaton}(q, a_i)$ is semantically equivalent to $\const{true}$ under the assumption that each $\tuple{q_k, i_k}$ is $\const{true}$, then 
    \[\begin{prooftree}
        \hypo{\tuple{q_1, i+ i_1} \in S^{\automaton}_{\word}}
        \hypo{\dots}
        \hypo{\tuple{q_{m}, i + i_{m}} \in S^{\automaton}_{\word}}
        \infer3{\tuple{q, i} \in S^{\automaton}_{\word}}
    \end{prooftree}.\]
For a \kl{2AFA} $\automaton$, the \kl{language} is defined as $\intro*\ljump{\automaton} \defeq \set{\word \in A^* \mid \tuple{\autosrc^{\automaton}, 1} \in S^{\automaton}_{{\lanchor}\word{\ranchor}}}$.
\AP We define the \intro*\kl(2AFA){size} $\intro*\autolen{\automaton}$ as $\sum_{\tuple{q, a} \in \univ{\automaton} \times (A \dcup \set{\lanchor, \ranchor})} \|\autotrans^{\automaton}(q, a)\|$, where $\|\fml[1]\|$ denotes the number of symbols occurring in the "positive boolean formula" $\fml[1]$.

By an on-the-fly exponential-size $1$-way NFA construction \cite[Lem.\ 1 and 5]{geffertTransformingTwoWayAlternating2014} for the \kl{language} $\ljump{\automaton[1]} \setminus \ljump{\automaton[2]}$ (see also \cite[Prop.\ 5.6]{nakamuraDerivativesGraphsPositive2025}),
we have the following proposition:
\begin{proposition}\label{proposition: 2AFA PSPACE}
    \AP
    The \intro*\kl{language inclusion problem} for \kl{2AFAs}---given a finite set $A$ and given two \kl{2AFAs} $\automaton[1]$ and $\automaton[2]$ over $A$, does $\ljump{\automaton[1]} \subseteq \ljump{\automaton[2]}$ hold?---is decidable in \textup{\kl{PSPACE}}.
\end{proposition}

\subsection{2AFAs construction}
Let $\const{c}_1, \const{c}_2, \dots$ be a set of pairwise distinct "elements".
We write $\intro*\RELk{k}$ for the class of all \kl{structures} $\struc$ such that $\univ{\struc} \subseteq \set{\const{c}_1, \dots, \const{c}_k}$.
Below, we consider \kl{path decompositions} of "width@@pathwidth" $k-1$ as "strings" over $\RELk{k}$.
Additionally, we only consider \emph{normalized} \kl{path decompositions}, based on \cite[\S 5.3.3 \& 5.3.4]{nakamuraDerivativesGraphsPositive2025}.
Let 
$\intro*\laNorm_{k} \defeq \RELk{k}^{+} \setminus (\laInac_{k} \cup \laIncon_{k})$, where%
\AP
\phantomintro{\laInac}\phantomintro{\laIncon}
\begin{align*}
\hspace{-1.em} \reintro*\laInac_{k} &\defeq \RELk{k}^* \compo \set*{\struc[1] \struc[2] \in \RELk{k}^{2} \mid \univ{\struc[1]} \cap \univ{\struc[2]} = \emptyset} \compo \RELk{k}^*,\\
\hspace{-1.em} \reintro*\laIncon_{k} &\defeq \RELk{k}^* \compo \set*{\struc[1] \struc[2] \in \RELk{k}^{2} \mid \exists b \in \vsig,\ b^{\struc[1]} \cap (\univ{\struc[1]} \cap \univ{\struc[2]})^{2} \neq b^{\struc[2]} \cap (\univ{\struc[1]} \cap \univ{\struc[2]})^{2}} \compo \RELk{k}^*.
\end{align*}
Here, $\compo$ denotes the "concatenation@@lang" of \kl{languages}.
$\laInac_{k}$ denotes the class of "path decompositions" having an inaccessible pair of "vertices" under the assumption that every pair of each "bag" is connected.
$\laIncon_{k}$ denotes the class of "path decompositions" having an inconsistent pair of "vertices" in that there is an
"edge" between the two "vertices" in one "bag" but not in the other "bag".
$\laNorm_{k}$ is sufficient to enumerate all \kl{structures} of \kl{pathwidth} at most $k-1$
in that $\REL_{\pw \le k-1}$ coincides with the isomorphism closure of $\set{\glue \mul{\struc} \mid \mul{\struc} \in \laNorm_{k}}$ \cite[Proposition 5.11]{nakamuraDerivativesGraphsPositive2025}.

Using \kl{2AFAs}, we can naturally encode the rules of \Cref{definition: decomposition}.
\begin{defi}\label{definition: loop-automaton to 2AFA}
    For $k \ge 2$ and a \kl{loop-automaton} $\automaton$ over $\vsig$,
    the \kl{2AFA} $\AP\intro*\loopautomatonAFA_{k}^{\automaton}$ over $\RELk{k}$ is defined as follows:
    \begin{itemize}
        \item \AP $\univ{\loopautomatonAFA_{k}^{\automaton}} \defeq \set{\intro*\Bautosrc} \dcup ((\set{\const{c}_1, \dots, \const{c}_k} \times \univ{\automaton})^2 \times \set{?, \checkmark})$; we abbreviate $\tuple{\tuple{x, p}, \tuple{y, q}, m}$ to $\tuple{\tuple{x, p}, \tuple{y, q}}_{m}$,
        \item $\autotrans^{\loopautomatonAFA_{k}^{\automaton}}(q, a) \defeq \bigvee \set{\fml[2] \mid \tuple{q, a} \leadsto \fml[2]}$ where
        $(\leadsto) \subseteq (\univ{\loopautomatonAFA_{k}^{\automaton}} \times (\RELk{k} \dcup \set{\lanchor,\ranchor}))
        \times \PBFML(\univ{\loopautomatonAFA_{k}^{\automaton}} \times \set{-1, 0, +1})$ is the minimal set closed under the following rules:
        \begin{gather*}
        \tuple{\Bautosrc, \lanchor} ~\leadsto~ \tuple{\tuple{\tuple{x, \src^{\automaton}}, \tuple{y, \tgt^{\automaton}}}_{?}, +1} \tag{I} \quad \text{for $x,y \in \set{\const{c}_1, \dots, \const{c}_k}$}\\
            \tuple{\tuple{\tilde{p}, \tilde{p}}_{\checkmark}, \struc} ~\leadsto~ \const{true} \quad\text{ if $\tilde{p} \in \univ{\struc} \times \univ{\automaton}$} \tag{R}\\
            \tuple{\tuple{\tilde{p}, \tilde{q}}_{\checkmark}, \struc} ~\leadsto~ \tuple{\tuple{\tilde{p}, \tilde{r}}_{\checkmark}, 0} \land \tuple{\tuple{\tilde{r}, \tilde{q}}_{\checkmark}, 0} \quad\text{ if $\tilde{r} \in \univ{\struc} \times \univ{\automaton}$} \tag{T}\\
            \tuple{\tuple{\tuple{x, p}, \tuple{y, q}}_{\checkmark}, \struc} ~\leadsto~ \const{true} \quad\text{ if $\tuple{p, q} \in \alpha^{\automaton} \mbox{ and } \tuple{x, y} \in \jump{\alpha}^{\struc}$ for $\alpha \in \vsig \dcup \set{\id}$} \tag{$\alpha$}\\
            \tuple{\tuple{\tuple{x, p}, \tuple{x, q}}_{\checkmark}, \struc} ~\leadsto~ \tuple{\tuple{\tuple{x, p'}, \tuple{x, q'}}_{\checkmark},0} \quad\text{ if $\tuple{p, q} \in \ell_{\tuple{p', q'}}^{\automaton}$} \tag{$\ell_{\tuple{p', q'}}$}\\
            \tuple{\tuple{\tilde{p}, \tilde{q}}_{\checkmark}, \struc} ~\leadsto~ \tuple{\tuple{\tilde{p}, \tilde{q}}_{?}, +1} \tag{$+1$}\\
            \tuple{\tuple{\tilde{p}, \tilde{q}}_{\checkmark}, \struc} ~\leadsto~ \tuple{\tuple{\tilde{p}, \tilde{q}}_{?}, -1} \tag{$-1$}\\
            \tuple{\tuple{\tilde{p}, \tilde{q}}_{?}, \struc} ~\leadsto~
            \tuple{\tuple{\tilde{p}, \tilde{q}}_{\checkmark}, 0} \quad\text{ if $\tilde{p}, \tilde{q} \in \univ{\struc} \times \univ{\automaton}$} \tag{$\checkmark$}
        \end{gather*}
        \item $\autosrc^{\loopautomatonAFA_{k}^{\automaton}} \defeq \Bautosrc$.
    \end{itemize}
\end{defi}
The rules are defined based on those in \Cref{definition: decomposition}, where
\begin{itemize}
    \item the rules ($+1$) and ($-1$) are used for moving vertices in the same quotient class;
    \item the check mark ``$\checkmark$'' in $\tuple{\tuple{x, p}, \tuple{y, q}}_{\checkmark}$ expresses that $x, y \in \univ{\struc}$ holds (where $\struc$ is the \kl{structure} in the current position);
    the question mark ``$?$'' in $\tuple{\tuple{x, p}, \tuple{y, q}}_{?}$ expresses that it has not yet been checked.
    They are introduced to check it after the rule ($+1$) or ($-1$).
\end{itemize}
By construction of $\loopautomatonAFA_{k}^{\automaton}$, we can show the following
(cf., \cite[Proposition 5.9]{nakamuraDerivativesGraphsPositive2025}).
\begin{proposition}\label{proposition: decomposition and 2AFA}
    Let $k \ge 2$ and let $\automaton$ be a \kl{loop-automaton} over $\vsig$.
    Let $\mul{\struc} = \struc_1 \dots \struc_n \in \RELk{k}^{+}$.
    For all $j \in \rangeone{n}$, $x, y \in \univ{\mul{\struc}\getith{j}}$, and $p, q \in \univ{\automaton}$, the following are equivalent:
    \begin{enumerate}
        \item \label{item: decomposition and 2AFA 1} $p\decomposedtransrel{\mul{\struc}}{\trace} q \mbox{ for some \kl{run} $\trace$ of $\automaton$ s.t.\ $\src^{\trace} = [\tuple{j,x}]_{\sim}$ and $\tgt^{\trace} = [\tuple{j,y}]_{\sim}$}$;
        \item \label{item: decomposition and 2AFA 2} $\tuple{\tuple{\tuple{x, p}, \tuple{y, q}}_{\checkmark}, j+1} \in S^{{\loopautomatonAFA_{k}^{\automaton}}}_{\lanchor \mul{\struc} \ranchor}$.
    \end{enumerate}
\end{proposition}
\begin{proof}
Both directions follow by straightforward transformations of the derivation trees:
the equally named rules correspond directly, while $(+1)$, $(-1)$, and $(\checkmark)$
only change the bag used to name the same quotient vertices.
\end{proof}
For a set $\la \subseteq \RELk{k}^{+}$,
we say that a \kl{loop-automaton} $\automaton$ is ""closed@@loopA"" on $\la$ if $\jump{\automaton}^{\glue \mul{\struc}} = \emptyset$ or $\jump{\automaton}^{\glue \mul{\struc}} = \univ{\glue \mul{\struc}}^2$ holds for every $\mul{\struc} \in \la$.
By \Cref{proposition: decomposition and 2AFA}, we have the following.
\begin{lem}\label{lemma: to 2AFA}
    For $k \ge 2$ and a \kl{loop-automaton} $\automaton$ over $\vsig$,
    we have:
    \[\ljump{\loopautomatonAFA_{k}^{\automaton}} ~=~ \set{\mul{\struc} \in \RELk{k}^{+} \mid \tuple{[\tuple{1,x}]_{\sim}, [\tuple{1,y}]_{\sim}} \in \jump{\automaton}^{\glue \mul{\struc}} \mbox{ for some $x, y \in \univ{\mul{\struc}\getith{1}}$}}.\]
    In particular, for a set $\la \subseteq \RELk{k}^{+}$, if $\automaton$ is "closed@@loopA" on $\la$,
    we have:
    \[\ljump{\loopautomatonAFA_{k}^{\automaton}} \cap \la ~=~ \set{\mul{\struc} \in \la \mid \jump{\automaton}^{\glue \mul{\struc}} \neq \emptyset}.\]
\end{lem}
\begin{proof}
    We have:
    \begin{align*}
        \hspace{-.8em}&\mul{\struc} \in \ljump{\loopautomatonAFA_{k}^{\automaton}}\\
        \hspace{-.8em}&\Leftrightarrow
        \exists x, y \in \univ{\mul{\struc}\getith{1}},\ \tuple{\tuple{\tuple{x, \src^{\automaton}}, \tuple{y, \tgt^{\automaton}}}_{\checkmark}, 2} \in S^{\loopautomatonAFA_{k}^{\automaton}}_{\lanchor \mul{\struc} \ranchor}
        \tag{By the form of $\loopautomatonAFA_{k}^{\automaton}$}\\
        \hspace{-.8em}&\Leftrightarrow \exists x, y \in \univ{\mul{\struc}\getith{1}}, \exists \trace \mbox{ s.t.\ $\src^{\trace} = \quo{\tuple{1, x}}_{\sim}$ and $\tgt^{\trace} = \quo{\tuple{1, y}}_{\sim}$},\ \src^{\automaton} \decomposedtransrel{\mul{\struc}}{\trace} \tgt^{\automaton} \tag{\Cref{proposition: decomposition and 2AFA}}\\
        \hspace{-.8em}&\Leftrightarrow \exists x, y \in \univ{\mul{\struc}\getith{1}}, \exists \trace \mbox{ s.t.\ $\src^{\trace} = \quo{\tuple{1, x}}_{\sim}$ and $\tgt^{\trace} = \quo{\tuple{1, y}}_{\sim}$},\ \src^{\automaton}  \transrel{\glue \mul{\struc}}{\trace} \tgt^{\automaton} \tag{\Cref{theorem: decomposition}}\\
        \hspace{-.8em}&\Leftrightarrow \exists x, y \in \univ{\mul{\struc}\getith{1}},\ \tuple{[\tuple{1, x}]_{\sim}, [\tuple{1, y}]_{\sim}} \in \jump{\automaton}^{\glue \mul{\struc}}. \tag{By definition}
    \end{align*}
    Hence, this completes the proof.
    (The second claim is immediate from the first one.)
\end{proof}

\subsection{Encoding the equational theory}\label{section: encoding equational theory}
Let $c_{\top}$ be a \kl{variable} for encoding the top element $\top$
and let $l$ and $r$ be \kl{variables} for indicating the left and right vertex, respectively.
We define
\AP
\phantomintro{\laTop}
\[\reintro*\laTop_{k} \defeq \RELk{k}^{*} \compo \set{\struc \in \RELk{k} \mid {\aterm[3]_{\top}^{\struc} \neq \univ{\struc}^2}} \compo \RELk{k}^{*}.\]
For every $\mul{\struc} \in \laNorm_{k} \setminus \laTop_{k}$,
we can encode the top element $\top$ by the \kl(loop-RKA){term} $\aterm[3]_{\top}^{*}$:
"ie", $\jump{c_{\top}^*}^{\glue \mul{\struc}} = \univ{\glue \mul{\struc}}^2$ holds.
Using them, we can encode the \kl{equational theory} as follows.
(Below, we recall $\loopRKAautomaton{\term}$ in \Cref{defi: term to automata}.)
\begin{lem}\label{lemma: encoding the equational theory of KL}
    Let $\term[1]$ and $\term[2]$ be \kl(loop-RKA){terms} over $\vsig$.
    We extend $\vsig$ with fresh variables $\aterm[3]_{\top}$, $l$, and $r$ that do not occur in $\term[1]$ or $\term[2]$ for the automata construction below.
    Then,
    \[\REL_{\pw \le k - 1} \klmodels \term[1] \le \term[2] \quad\iff\quad \ljump{\loopautomatonAFA_{k}^{\loopRKAautomaton{\aterm[3]_{\top}^* {l} \term[1] {r} \aterm[3]_{\top}^*}}} \subseteq \ljump{\loopautomatonAFA_{k}^{\loopRKAautomaton{\aterm[3]_{\top}^* {l} \term[2] {r} \aterm[3]_{\top}^*}}} \cup \laInac_{k} \cup \laIncon_{k} \cup \laTop_{k}.\]
\end{lem}
\begin{proof}
    We have:
    \begin{align*}
        &\REL_{\pw \le k - 1} \klmodels \term[1] \le \term[2] 
        ~\Leftrightarrow~ \set{\glue \mul{\struc[1]} \mid \mul{\struc[1]} \in \laNorm_{k}} \klmodels \term[1] \le \term[2] \\
        &~\Leftrightarrow~ \set{\glue \mul{\struc[1]} \mid \mul{\struc[1]} \in \laNorm_{k} \setminus \laTop_{k}} \klmodels \term[1] \le \term[2] \tag{$c_{\top}$ is "fresh"}\\
        &~\Leftrightarrow~ \set{\glue \mul{\struc[1]} \mid \mul{\struc[1]} \in \laNorm_{k} \setminus \laTop_{k}}  \klmodels c_{\top}^* {l} \term[1] {r} c_{\top}^* \le c_{\top}^* {l} \term[2] {r} c_{\top}^* \tag{$c_{\top}^*$ expresses $\top$; $l$, $r$ are "fresh"}\\
        &~\Leftrightarrow~ \forall \mul{\struc[1]} \in \laNorm_{k} \setminus \laTop_{k},\ \jump{\loopRKAautomaton{c_{\top}^* {l} \term[1] {r} c_{\top}^*}}^{\glue \mul{\struc[1]}} \subseteq \jump{\loopRKAautomaton{c_{\top}^* {l} \term[2] {r} c_{\top}^*}}^{\glue \mul{\struc[1]}} \tag{\Cref{prop: automata}}\\
        &~\Leftrightarrow~ \set{\mul{\struc[1]} \mid \jump{\loopRKAautomaton{c_{\top}^* {l} \term[1] {r} c_{\top}^*}}^{\glue \mul{\struc[1]}} \neq \emptyset} \cap (\laNorm_{k} \setminus \laTop_{k}) \subseteq \set{\mul{\struc[1]} \mid \jump{\loopRKAautomaton{c_{\top}^* {l} \term[2] {r} c_{\top}^*}}^{\glue \mul{\struc[1]}} \neq \emptyset} \tag{$\loopRKAautomaton{c_{\top}^* {l} \term[1] {r} c_{\top}^*}$ and $\loopRKAautomaton{c_{\top}^* {l} \term[2] {r} c_{\top}^*}$ are "closed@@loopA" on $\laNorm_{k} \setminus \laTop_{k}$}\\
        &~\Leftrightarrow~ \ljump{\loopautomatonAFA_{k}^{\loopRKAautomaton{\aterm[3]_{\top}^* {l} \term[1] {r} \aterm[3]_{\top}^*}}} \cap (\laNorm_{k} \setminus \laTop_{k}) \subseteq \ljump{\loopautomatonAFA_{k}^{\loopRKAautomaton{\aterm[3]_{\top}^* {l} \term[2] {r} \aterm[3]_{\top}^*}}} \tag{\Cref{lemma: to 2AFA}}\\
        &~\Leftrightarrow~ \ljump{\loopautomatonAFA_{k}^{\loopRKAautomaton{\aterm[3]_{\top}^* {l} \term[1] {r} \aterm[3]_{\top}^*}}} \subseteq \ljump{\loopautomatonAFA_{k}^{\loopRKAautomaton{\aterm[3]_{\top}^* {l} \term[2] {r} \aterm[3]_{\top}^*}}} \cup \laInac_{k} \cup \laIncon_{k} \cup \laTop_{k}.
    \end{align*}
    Hence, this completes the proof.
\end{proof}

\subsection{Binary encoding of structures}\label{section: binary encoding}
The cardinality of $\RELk{k}$ is approximately $2^{\card \vsig \cdot k^2}$.
Due to this, the \kl{size} $\autolen{\loopautomatonAFA_{k}^{\automaton}}$ is exponential in $k$.
Nevertheless, we can avoid this exponential blowup by a standard binary encoding of \kl{structures} (cf., e.g., \cite[Definition 2.1]{immermanDescriptiveComplexity1999}).

Let $\vsig = \set{a_1, \dots, a_m}$.
Given $k \ge 2$ and a \kl{structure} $\struc \in \RELk{k}$,
we define the following ($k + m k^2$)-bit \kl{string} $\intro*\binenc{\struc}{k}$:
\begin{itemize}
    \item The first $k$ bits encode the universe $\univ{\struc} \subseteq \set{\const{c}_1,\dots,\const{c}_k}$:
    the $i$-th bit is $1$ if $\const{c}_i \in \univ{\struc}$, and $0$ otherwise.

    \item For each $\ell$ from $1$ to $m$,
    the next $k^2$ bits encode the relation $a_\ell^{\struc} \subseteq \univ{\struc}^2$:
    the $((p-1) k + q)$-th bit is $1$ if $\tuple{\const{c}_p, \const{c}_q} \in a_\ell^{\struc}$, and $0$ otherwise.
\end{itemize}
For instance, if $\struc$ is given as follows, $k = 3$, and $\vsig = \set{a, b}$,
then
$\binenc{\struc}{k}$ is defined as follows:
\[
\struc ~=~ \begin{tikzpicture}[baseline = -.5ex]
    \graph[grow right = 1.cm, branch down = 2.5ex]{
    {s1/{$\const{c}_1$}[vert]} -!- {t1/{$\const{c}_2$}[vert]}
    };
    \graph[use existing nodes, edges={color=black, pos = .5, earrow}, edge quotes={fill=white, inner sep=1pt,font= \scriptsize}]{
        s1 ->["$a$", loop above, out = 60, in = 120, looseness = 5] s1;
        s1 ->["$a$", bend right] t1;
        t1 ->["$b$", bend right] s1;
    };
\end{tikzpicture} \hspace{1em} \leadsto \hspace{1em} \arraycolsep=3pt\begin{array}{c ccc|ccc ccc ccc|ccccccccc}
   & \multicolumn{3}{c|}{\univ{\struc}} & \multicolumn{9}{c|}{a^{\struc}} & \multicolumn{9}{c}{b^{\struc}}\\
    \cline{2-22}
   \binenc{\struc}{k} ~=~ & 1 & 1 & 0 & 1 & 1 & 0 & 0 & 0 & 0 & 0 & 0 & 0 & 0 & 0 & 0 & 1 & 0 & 0 & 0 & 0 & 0
\end{array}\]
For a sequence $\mul{\struc} = \struc_1 \dots \struc_n \in \RELk{k}^{+}$,
set $\binenc{\mul{\struc}}{k} \defeq \binenc{\struc_1}{k} \dots \binenc{\struc_n}{k}$.
For a set $\mathcal{L} \subseteq \RELk{k}^{+}$,
set $\binenc{\mathcal{L}}{k} \defeq \set{\binenc{\mul{\struc}}{k} \mid \mul{\struc} \in \mathcal{L}}$.
We can redefine the construction of \Cref{definition: loop-automaton to 2AFA} via this binary encoding.
\begin{lemma}\label{lemma: to binary AFA}
    Given $k \ge 2$ and a \kl{loop-automaton} $\automaton$ over $\vsig$,
    we can construct a \kl{2AFA} $\AP\intro*\binloopautomatonAFA_{k}^{\automaton}$ over the alphabet $\set{0,1}$
    such that $\ljump{\binloopautomatonAFA_{k}^{\automaton}} = \binenc{\ljump{\loopautomatonAFA_{k}^{\automaton}}}{k}$,
    in time polynomial in $k$, $\card\vsig$, and the number of states and transitions of $\automaton$.
\end{lemma}
\begin{proof}[Proof Sketch]
    We modify the construction of \Cref{definition: loop-automaton to 2AFA}
    using the binary encoding above.
    For instance, for the rule $(a_i)$,
    when $\tuple{p, q} \in a_i^{\automaton}$ and $x,y\in\rangeone{k}$,
    we replace the rule
    \[\tuple{\tuple{\tuple{\const{c}_x, p}, \tuple{\const{c}_y, q}}_{\checkmark}, \struc} ~\leadsto~ \const{true} \quad\text{ if $\tuple{\const{c}_x, \const{c}_y} \in a_i^{\struc}$}\]
    with the rules so that
    we move the head to the position of the bit corresponding to $\tuple{\const{c}_x, \const{c}_y} \in a_i^{\struc}$ and check whether it is $1$;
    the simulated state for the current \kl{structure} is kept at the first bit of its block, so the $P$-th bit is reached after $P-1$ moves;
    more precisely, as follows:
    \begin{align*}
        \tuple{\tuple{\tuple{\const{c}_x, p}, \tuple{\const{c}_y, q}}_{\checkmark}, b} &~\leadsto~ \tuple{\tuple{\tuple{\const{c}_x, p}, \tuple{\const{c}_y, q}}_{\checkmark}^{(a_i), 0}, 0},\\
        \tuple{\tuple{\tuple{\const{c}_x, p}, \tuple{\const{c}_y, q}}_{\checkmark}^{(a_i), \ell-1}, b} &~\leadsto~ \tuple{\tuple{\tuple{\const{c}_x, p}, \tuple{\const{c}_y, q}}_{\checkmark}^{(a_i), \ell}, +1} \text{ for $\ell \in \rangeone{k + m k^2 - 1}$,}\\
        \tuple{\tuple{\tuple{\const{c}_x, p}, \tuple{\const{c}_y, q}}_{\checkmark}^{(a_i), \ell}, 1} &~\leadsto~ \const{true} \tag*{if $\ell = k + (i-1)k^2 + (x-1)k + y - 1$,}
    \end{align*}
    where $b \in \set{0, 1}$.
    The original moves $(+1)$ and $(-1)$ between adjacent \kl{structures} are simulated by moving the head by exactly $k + m k^2$ bit positions to the next or previous block, using intermediate counter states.
    All the other rules are handled similarly, using polynomially many intermediate states.
    Additionally, putting $N \defeq k + m k^2$,
    we can also construct in polynomial time a \kl{2AFA} recognizing the \kl{language}
    \[
        \binenc{\RELk{k}^{+}}{k} = (\set{0,1}^{N} \setminus (\mathcal{B}_{\emptyset} \cup \mathcal{B}_{\mathrm{edge}}))^{+},
    \]
    where
    \[
        \mathcal{B}_{\emptyset} \defeq \set{0}^{k}\set{0,1}^{mk^2}
    \]
    is the set of blocks whose universe is empty, and
    \begin{align*}
        \mathcal{B}_{\mathrm{edge}}
        \defeq
        \bigcup_{\substack{j \in \rangeone{m}\\ x,y \in \rangeone{k}}}
        \Big(&
        \set{0,1}^{x-1}0
        \set{0,1}^{k-x+(j-1)k^2+(x-1)k+y-1}
        1
        \set{0,1}^{mk^2-((j-1)k^2+(x-1)k+y)}
        \\
        {}\cup{}&
        \set{0,1}^{y-1}0
        \set{0,1}^{k-y+(j-1)k^2+(x-1)k+y-1}
        1
        \set{0,1}^{mk^2-((j-1)k^2+(x-1)k+y)}
        \Big)
    \end{align*}
    is the set of blocks in which some relation bit uses a vertex outside the universe.
    The \kl{2AFA} first checks that the input is a non-empty concatenation of blocks of length $N$.
    It excludes $\mathcal{B}_{\emptyset}$ by checking, in every block, that at least one of the first $k$ bits is $1$.
    It excludes $\mathcal{B}_{\mathrm{edge}}$ by universally branching over relation-bit positions:
    if the current relation bit is $0$, the branch accepts immediately;
    if it is $1$, the automaton moves within the same block to the corresponding universe bits $b_x$ and $b_y$ and checks that both are $1$.
    By taking the intersection of them, we obtain the desired \kl{2AFA}.
\end{proof}

The same bit checks give the auxiliary automata used below.
\begin{lemma}\label{lemma: binary auxiliary languages}
Given $k$ and $\vsig=\set{a_1,\dots,a_m}$ (including the fresh variable $c_{\top}$),
we can construct, in time polynomial in $k$ and $m$, \kl{2AFAs}
$\automaton_k^{\mathrm{Inac}}$, $\automaton_k^{\mathrm{Incon}}$, and
$\automaton_k^{\mathrm{Top}}$ over $\set{0,1}$ such that
$\ljump{\automaton_k^{\mathrm{Inac}}}=\binenc{\laInac_k}{k}$,
$\ljump{\automaton_k^{\mathrm{Incon}}}=\binenc{\laIncon_k}{k}$,
$\ljump{\automaton_k^{\mathrm{Top}}}=\binenc{\laTop_k}{k}$.
\end{lemma}
\begin{proof}
Put $N\defeq k+mk^2$.
First construct property-checking \kl{2AFAs} on strings of $N$-bit blocks.
For $\automaton_k^{\mathrm{Inac}}$, guess two adjacent blocks and universally verify,
for every $x\in\rangeone{k}$, that $\const{c}_x$ does not occur in both universes.
For $\automaton_k^{\mathrm{Incon}}$, guess two adjacent blocks and $a_i,x,y$,
verify that both endpoints occur in both universes, and verify that the two bits differ for
$a_i(\const{c}_x,\const{c}_y)$.
For $\automaton_k^{\mathrm{Top}}$, guess a block and $x,y$ whose universe bits are $1$
and verify that the corresponding $c_{\top}$-bit is $0$.
Each construction has polynomially many states and transition-formula symbols:
the block positions are reached with counters of length at most $N$, and the choices range over
at most $m k^2$ triples.
Each automaton also checks that the input is a non-empty sequence of blocks encoding
structures, using the same checks as in the proof of \Cref{lemma: to binary AFA}.
The resulting automata satisfy the displayed language equalities.
\end{proof}

\begin{theorem}\label{theorem: KL PSPACE-complete restatement}
    The (in)\kl{equational theory} of \kl{loop-RKA}---given a finite set $\vsig$ and two \kl{loop-RKA terms} $\term[1]$ and $\term[2]$ over $\vsig$, does $\REL \klmodels \term[1] \le \term[2]$ hold?---is \textup{\kl{PSPACE}}-complete.
\end{theorem}
\begin{proof}
    \proofcase{(Lower bound)}
    The \kl{equational theory} of \kl{RKA}, which is equivalent to the "language equivalence problem" of "regular expressions" (\Cref{footnote: lang and REL equiv}), is already \kl{PSPACE}-hard \cite{meyerEquivalenceProblemRegular1972}.
    \proofcase{(Upper bound)}
    By letting $k \defeq \iw(\term[1]) + 1$, we have: 
    \begin{align*}
    \REL \klmodels \term[1] \le \term[2]
    &~\Leftrightarrow~ \REL_{\pw \le k-1} \klmodels \term[1] \le \term[2] \tag{\Cref{proposition: bounded pw property}}\\
    &~\Leftrightarrow~ \ljump{\loopautomatonAFA_{k}^{\loopRKAautomaton{\aterm[3]_{\top}^* {l} \term[1] {r} \aterm[3]_{\top}^*}}} \subseteq \ljump{\loopautomatonAFA_{k}^{\loopRKAautomaton{\aterm[3]_{\top}^* {l} \term[2] {r} \aterm[3]_{\top}^*}}} \cup \laInac_{k} \cup \laIncon_{k} \cup \laTop_{k} \tag{\Cref{lemma: encoding the equational theory of KL}}\\
    &~\Leftrightarrow~ \ljump{\binloopautomatonAFA_{k}^{\loopRKAautomaton{\aterm[3]_{\top}^* {l} \term[1] {r} \aterm[3]_{\top}^*}}} \subseteq \ljump{\binloopautomatonAFA_{k}^{\loopRKAautomaton{\aterm[3]_{\top}^* {l} \term[2] {r} \aterm[3]_{\top}^*}}}
    \cup \binenc{\laInac_{k} \cup \laIncon_{k} \cup \laTop_{k}}{k} \tag{\Cref{lemma: to binary AFA}}.
    \end{align*}
    By \Cref{lemma: binary auxiliary languages}, we can construct \kl{2AFAs} for
    $\binenc{\laInac_{k}}{k}$, $\binenc{\laIncon_{k}}{k}$, and $\binenc{\laTop_{k}}{k}$ in polynomial time.
    Thus we can construct a \kl{2AFA} for the right-hand side "language" by taking the union of them.
    Hence, by \Cref{proposition: 2AFA PSPACE}, this completes the proof.
\end{proof}
The upper bound of \Cref{theorem: KL PSPACE-complete} immediately follows from \Cref{theorem: KL PSPACE-complete restatement}.
The lower bound of \Cref{theorem: KL PSPACE-complete} also follows from the same proof as in \Cref{theorem: KL PSPACE-complete restatement}.

\section{Loop-RKA with Additional Operators}\label{section: encoding extras}
In this section,
we briefly give some encodings of additional operators in our automata construction, based on \cite{nakamuraDerivativesGraphsPositive2025}.
For combined restrictions below, comma-separated superscripts denote intersection; for example,
$\REL^{\resTest,\resNom}\defeq\REL^{\resTest}\cap\REL^{\resNom}$.
\paragraph*{For top}
As in \Cref{section: encoding equational theory}, introduce a fresh variable $\aterm[3]_{\top}$
and replace each native occurrence of $\top$ with $\aterm[3]_{\top}^{*}$.
For a class $\algclass\subseteq\REL$, let
$\algclass^{\intro*\resTop}\defeq
\set{\struc\in\algclass\mid \jump{\aterm[3]_{\top}^{*}}^{\struc}=\univ{\struc}^{2}}$.
Recall that $\laTop_k$ rejects a word if $\aterm[3]_{\top}$ is not the full relation in some bag.
If $\mul{\struc} \notin \laTop_k \cup \laInac_k$, then consecutive bags overlap and
\[
\jump{\aterm[3]_{\top}^{*}}^{\glue\mul{\struc}}=\univ{\glue\mul{\struc}}^2.
\tag{top}\label{class cond: top}
\]
Conversely, because $\aterm[3]_{\top}$ is fresh, any normalized path decomposition over the original signature
can be expanded by interpreting $\aterm[3]_{\top}$ as the full relation in each bag.
After forgetting this fresh relation, the glued structure is unchanged.
Thus this stronger bag-wise condition preserves the validity of terms in which $\aterm[3]_{\top}$ occurs only as $\aterm[3]_{\top}^{*}$;
it is not asserted to characterize all structures satisfying \eqref{class cond: top} with an arbitrary interpretation of $\aterm[3]_{\top}$.
For $\binenc{\laTop_{k}}{k}$, the probe construction of
\Cref{lemma: binary auxiliary languages} gives a polynomial-size \kl{2AFA}.
Hence, we can encode "top".

\paragraph*{For tests}
Let $B, \bar{B} \subseteq \vsig$ be two disjoint sets of \kl{variables}
with a bijection $\bar{\mbox{\quad}} \colon B \to \bar{B}$.
Before translation, test expressions are Boolean formulas
$\varphi ::= b \mid \mathsf{true} \mid \mathsf{false} \mid
\varphi\land\varphi \mid \varphi\lor\varphi \mid \lnot\varphi$ for $b\in B$,
interpreted as subidentities.
To encode \intro*\kl{tests} as in \kl{KAT} \cite{kozenKleeneAlgebraTests1996},
we require each pair $b,\bar b$ to form a partition of the identity relation:
\begin{align*}
b^{\struc} \cap \bar b^{\struc} &= \emptyset,
& b^{\struc} \cup \bar b^{\struc} &= \diagonal_{\univ{\struc}}
\qquad\text{for every $b \in B$.} \tag{test}\label{class cond: test}
\end{align*}
\AP
For a class $\algclass \subseteq \REL$,
let $\algclass^{\intro*\resTest} \defeq \set{\struc \in \algclass \mid \text{$\struc$ satisfies \eqref{class cond: test}}}$.
In this class, for test terms $p$ generated from variables in $B \cup \bar{B}$ and constants $\id,\emp$ using $\compo$ and $\union$,
we define the complemented term $p^{-}$ (w.r.t.\  the identity relation) as follows:
\begin{gather*}
    b^{-} \defeq \bar{b}, \quad
    \bar{b}^{-} \defeq b, \quad
    \id^{-} \defeq \emp, \quad
    \emp^{-} \defeq \id, \quad
    (p \compo q)^{-} \defeq p^{-} \union q^{-}, \quad
    (p \union q)^{-} \defeq p^{-} \compo q^{-}.
\end{gather*}
Then $\jump{p^{-}}^{\struc} = \diagonal_{\univ{\struc}} \setminus \jump{p}^{\struc}$ holds for all $\struc \in \REL^{\resTest}$.
We thus can encode \kl{tests} (propositional formulas),
by expressing true, false, conjunction, disjunction,
as $\id$, $\emp$, $\compo$, and $\union$,
respectively,
and expressing complement as above.
Using \kl{tests}, for instance, we can encode propositional while programs \cite{fischerPropositionalDynamicLogic1979,kozenKleeneAlgebraTests1997}:
\begin{align*}
    \mbox{\textbf{while }} p \mbox{\textbf{ do }} \term[1] &\;\defeq\; (p \term[1])^{*} p^{-}, &
    \mbox{\textbf{if }} p \mbox{\textbf{ then }} \term[1] \mbox{\textbf{ else }} \term[2] &\;\defeq\; (p \term[1]) \union (p^{-} \term[2]).
\end{align*}
To encode the condition \eqref{class cond: test}, we define the following \kl{language}:
\AP
\phantomintro{\laTest}
\[\reintro*\laTest_{k} \;\defeq\;
\RELk{k}^{*} \compo
\set*{\struc \in \RELk{k} \mid \exists b \in B,\
b^{\struc} \cap \bar{b}^{\struc} \neq \emptyset
\text{ or } b^{\struc} \cup \bar{b}^{\struc} \neq \diagonal_{\univ{\struc}} } \compo
\RELk{k}^{*}.\]
Then, up to structure isomorphism,
$\REL^{\resTest}_{\pw \le k-1}$ coincides with
$\set{\glue \mul{\struc} \mid \mul{\struc} \in \RELk{k}^{+} \setminus (\laTest_{k} \cup \laIncon_{k})}$.
For the language $\binenc{\laTest_{k}}{k}$,
a \kl{2AFA} guesses a block, $b\in B$, and endpoint indices, and then probes either
a pair belonging to both $b$ and $\bar b$, an off-diagonal pair in their union,
or a missing diagonal pair.
Combined with \Cref{lemma: to binary AFA}, this gives in polynomial time
a \kl{2AFA} of size polynomial in $k$ and $\card\vsig$.
Hence, we can encode \kl{tests}.

\paragraph*{For converse}
Let $C, \breve{C} \subseteq \vsig$ be two disjoint sets of \kl{variables} with a bijection $\breve{\mbox{\quad}} \colon C \to \breve{C}$.
To encode \intro*\kl{converse}, it suffices to force the class of \kl{structures} $\struc$
so that
\begin{align*}
\breve{c}^{\struc} = (c^{\struc})^{\smile} \text{\quad for every $c \in C$.} \tag{converse}\label{class cond: converse}
\end{align*}
\AP
For a class $\algclass \subseteq \REL$,
let $\algclass^{\intro*\resConv} \defeq \set{\struc \in \algclass \mid \text{$\struc$ satisfies \eqref{class cond: converse}}}$.
In this class, converse is pushed to atoms.
For each term $\term$, the resulting term, also denoted by $\term^{\smile}$, satisfies
$\jump{\term^{\smile}}^{\struc} = (\jump{\term}^{\struc})^{\smile}$ for all
$\struc \in \REL^{\resConv,\resTest,\resNom}$, by the following rules:
\begin{gather*}
    c^{\smile} \defeq \breve{c}, \quad
    \breve{c}^{\smile} \defeq c, \quad
    \id^{\smile} \defeq \id, \quad
    \emp^{\smile} \defeq \emp, \quad
    (\term[1] \compo \term[2])^{\smile} \defeq \term[2]^{\smile} \compo \term[1]^{\smile}, \\
    (\term[1] \union \term[2])^{\smile} \defeq \term[1]^{\smile} \union \term[2]^{\smile}, \quad
    (\term[1]^{*})^{\smile} \defeq (\term[1]^{\smile})^{*}, \quad
    (\term[1]^{\lop})^{\smile} \defeq (\term[1]^{\smile})^{\lop}, \quad
    (\term[1]^{\smile})^{\smile} \defeq \term[1].
\end{gather*}
To encode the condition \eqref{class cond: converse}, we define the following \kl{language}:
\AP
\phantomintro{\laConv}
\[\reintro*\laConv_{k} \;\defeq\;
\RELk{k}^{*} \compo
\set{\struc \in \RELk{k} \mid \exists c \in C,\ (c^{\struc})^{\smile} \neq \breve{c}^{\struc} } \compo 
\RELk{k}^{*}.\]
Then, up to structure isomorphism,
$\REL^{\resConv}_{\pw \le k-1}$ coincides with
$\set{\glue \mul{\struc} \mid \mul{\struc} \in \RELk{k}^{+} \setminus \laConv_{k}}$.
For the \kl{language} $\binenc{\laConv_{k}}{k}$,
a \kl{2AFA} guesses a block, $c\in C$, and $x,y$, and compares the bits for
$c(\const{c}_x,\const{c}_y)$ and $\breve c(\const{c}_y,\const{c}_x)$.
Together with \Cref{lemma: to binary AFA}, this is a polynomial-time construction
of size polynomial in $k$ and $\card\vsig$.
Hence, we can encode \kl{converse}.

\paragraph*{For nominals}
To encode \intro*\kl{nominals} (from hybrid modal logic) \cite{arecesHybridLogics2007}, it suffices to force the class of \kl{structures} $\struc$
so that
\begin{align*}
\forall l \in L,\ \exists x \in \univ{\struc},\quad
l^{\struc} = \set{(x,x)}. \tag{nominal}\label{class cond: nominal}
\end{align*}
\AP
For a class $\algclass \subseteq \REL$,
let $\algclass^{\intro*\resNom} \defeq \set{\struc \in \algclass \mid \text{$\struc$ satisfies \eqref{class cond: nominal}}}$.
To encode the condition \eqref{class cond: nominal}, we define the following \kl{language}:
\AP
\phantomintro{\laNom}
\begin{align*}
    \reintro*\laNom_{k} &\defeq
    \bigcup_{l \in L}
    \left(\begin{aligned}
       & \set{\struc \in \RELk{k} \mid l^{\struc} = \emptyset}^{+} \cup {} \\
       & (\RELk{k}^{*} \compo \set{\struc \in \RELk{k} \mid l^{\struc} \not\subseteq \diagonal_{\univ{\struc}}} \compo \RELk{k}^{*}) \cup {} \\
       & \RELk{k}^{*} \compo \set*{
            \begin{aligned}
                &\struc_1 \dots \struc_{n} \in \RELk{k}^{+} \mid \exists x, \exists y, \tuple{x, x} \in l^{\struc_1}   \\
                &{} \land \tuple{y, y} \in l^{\struc_n} \land  (x \neq y \lor \exists j \in \rangeone{n}, x \not\in \univ{\struc_{j}})
            \end{aligned}} \compo \RELk{k}^{*}
    \end{aligned}\right).
\end{align*}
Then, up to structure isomorphism, $\REL^{\resNom}_{\pw \le k-1}$ coincides with
$\set{\glue \mul{\struc} \mid \mul{\struc} \in \RELk{k}^{+} \setminus \laNom_{k}}$.
The three alternatives in $\laNom_k$ exclude, respectively, an empty interpretation, a non-diagonal pair, and more than one quotient vertex.
For the language $\binenc{\laNom_{k}}{k}$,
a \kl{2AFA} guesses $l\in L$ and one of the three alternatives above.
It scans all blocks for the first alternative, probes an off-diagonal pair for the second,
and stores $x,y\in\rangeone{k}$ while scanning the selected factor for the third.
Combining this with \Cref{lemma: to binary AFA} gives a polynomial-time construction
of size polynomial in $k$ and $\card\vsig$.
Hence, we can encode \kl{nominals}.

\paragraph*{Putting it all together}
We finally consider \kl{loop-RKA} with \kl{tests}, \kl{converse}, and \kl{nominals};
precisely, we consider \kl{loop-RKA} over the class
$\REL^{\resTop, \resConv, \resTest, \resNom}$ via the above encodings for these extensions,
where $B, \bar{B}, C, \breve{C}, L$ are pairwise disjoint subsets of $\vsig$.
We recall the fact that 
\kl{loop-RKA} with \kl{tests}, \kl{converse}, and \kl{nominals}
still has the \kl{linearly bounded pathwidth model property}:
\begin{proposition}[{\cite[Proposition 6.5]{nakamuraDerivativesGraphsPositive2025}}]\label{proposition: ext bounded pw property}
    For all \kl{loop-RKA terms} $\term[1]$ and $\term[2]$,
    we have:
    \[\REL^{\resTop, \resConv, \resTest, \resNom} \klmodels \term[1] \le \term[2] \quad\Leftrightarrow\quad
    \REL^{\resTop, \resConv, \resTest, \resNom}_{\pw \le \iw(\term[1]) + \card L} \klmodels \term[1] \le \term[2].\]
\end{proposition}

\begin{theorem*}[Restatement of \Cref{theorem: KL extras PSPACE-complete}]
\theoremKLextrasPSPACEcomplete
\end{theorem*}
\begin{proof}
\proofcase{(Lower bound)}
Immediate from the lower bound of \Cref{theorem: KL PSPACE-complete}.

\proofcase{(Upper bound)}
Given two \kl{loop-RKA} terms with "top", "tests", "converse", and "nominals", $\term[1]$ and $\term[2]$,
choose fresh variables $\aterm[3]_{\top}$, $l$, and $r$ that do not occur in $\term[1]$ or $\term[2]$ and are disjoint from $B, \bar{B}, C, \breve{C}$, and $L$.
Let $\term[1]'$ and $\term[2]'$ be the \kl{loop-RKA} terms obtained from $\term[1]$ and $\term[2]$ by replacing "top", "tests", and "converse" with their encodings as above, respectively.
By letting $k \defeq \iw(\term[1]') + \card L + 1$, we have: 
\begin{align*}
&\REL^{\resTop, \resConv, \resTest, \resNom}\klmodels \term[1] \le \term[2]
~\Leftrightarrow~ \REL^{\resTop, \resConv, \resTest, \resNom}\klmodels \term[1]' \le \term[2]'\\
&\Leftrightarrow~ \REL^{\resTop, \resConv, \resTest, \resNom}_{\pw \le k-1} \klmodels \term[1]' \le \term[2]' \tag{\Cref{proposition: ext bounded pw property}}\\
&\Leftrightarrow~ \ljump{\loopautomatonAFA_{k}^{\loopRKAautomaton{\aterm[3]_{\top}^* {l} \term[1]' {r} \aterm[3]_{\top}^*}}} \subseteq \ljump{\loopautomatonAFA_{k}^{\loopRKAautomaton{\aterm[3]_{\top}^* {l} \term[2]' {r} \aterm[3]_{\top}^*}}} \cup \laInac_{k} \cup \laIncon_{k} \cup \laTop_{k}
\cup \laConv_{k}
\cup \laTest_{k}
\cup \laNom_{k}
\tag{\Cref{lemma: encoding the equational theory of KL} with the encodings above}\\
&\Leftrightarrow~ \ljump{\binloopautomatonAFA_{k}^{\loopRKAautomaton{\aterm[3]_{\top}^* {l} \term[1]' {r} \aterm[3]_{\top}^*}}} \subseteq \ljump{\binloopautomatonAFA_{k}^{\loopRKAautomaton{\aterm[3]_{\top}^* {l} \term[2]' {r} \aterm[3]_{\top}^*}}}
\cup \binenc{
\laInac_{k}
\cup \laIncon_{k}
\cup \laTop_{k}
\cup \laConv_{k}
\cup \laTest_{k}
\cup \laNom_{k}}{k}
\tag{\Cref{lemma: to binary AFA}}.
\end{align*}
For the right-hand side language,
since we can construct the \kl{2AFA} for each \kl{language} in polynomial time, we can construct its \kl{2AFA} by taking the union.
Hence, by \Cref{proposition: 2AFA PSPACE}, this completes the proof.
\end{proof}

\Cref{corollary: KA with domain PSPACE-complete} immediately follows from \Cref{theorem: KL extras PSPACE-complete}
via the encoding of \kl{domain} \eqref{equation: dom range by loop top}.

\bibliography{main}

\appendix

\end{document}